\documentclass[11pt]{article}
\usepackage{enumerate, enumitem, hyperref}
\usepackage{latexsym}
\usepackage{amsmath}
\usepackage{amssymb}
\usepackage{amsthm}
\usepackage{epsfig}
\usepackage{bbm}
\usepackage[latin1]{inputenc}
\usepackage{color}

\def\namedlabel#1#2{\begingroup
    #2%
    \def\@currentlabel{#2}%
    \phantomsection\label{#1}\endgroup
}

\usepackage{tikz}
\usetikzlibrary{shapes}
\usetikzlibrary{patterns}
\usetikzlibrary{arrows,positioning,decorations.pathmorphing,trees}

\usepackage{fullpage}

\newtheorem{theorem}{Theorem}[section]

\theoremstyle{remark}

\newenvironment{restatetheorem}[1]
{\begingroup
  \def\thetheorem{\ref{#1}}
  \begin{theorem}}
{\end{theorem}
 \addtocounter{theorem}{-1}
 \endgroup}

\usepackage{natbib}

\newtheorem{innercustomgeneric}{\customgenericname}
\providecommand{\customgenericname}{}
\newcommand{\newcustomtheorem}[2]{%
  \newenvironment{#1}[1]
  {%
   \renewcommand\customgenericname{#2}%
   \renewcommand\theinnercustomgeneric{##1}%
   \innercustomgeneric
  }
  {\endinnercustomgeneric}
}

\newcustomtheorem{customthm}{Theorem}
\newcustomtheorem{customlemma}{Lemma}
\newcustomtheorem{customcor}{Corollary}

\usepackage[noend]{algorithmic}
\usepackage{algorithm}
\floatname{algorithm}{Procedure}

\title{A Recursive Measure of Voting Power\\ that Satisfies Reasonable Postulates\\[.2cm]}
\author{A. Abizadeh\thanks{Department of Political Science, McGill University: {\tt arash.abizadeh@mcgill.ca}} 
\and A. Vetta\thanks{Department of Mathematics \& Statistics and School of Computer Science, McGill University: {\tt adrian.vetta@mcgill.ca}}}

\begin{document}

\maketitle

\begin{abstract}
We design a recursive measure of voting power based on partial as well as full voting efficacy.
Classical measures, by contrast, incorporate solely full efficacy.
We motivate our design by representing voting games using a division lattice
and via the notion of random walks in stochastic processes,
and show the viability of our recursive measure by proving it satisfies a plethora of postulates that any reasonable voting measure should satisfy.
These include the iso-invariance, dummy, dominance, donation, minimum-power bloc, and quarrel postulates.
\end{abstract}


\section{Introduction}\label{sec:intro}

There have been two approaches to justifying proposed measures of voting power.
The first seeks to identify a set of reasonable axioms that uniquely pick out a single measure of voting power.
To date this {\em axiomatic} approach has proved a failure:
while many have provided axiomatic characterizations of various measures,
that is, the set of axioms the measures uniquely satisfy,
no one has done so for a set of axioms all of which are independently justified.
In other words, no one has succeeded in showing why it would be reasonable to expect a measure of voting power
to satisfy the entire set of axioms that uniquely pick out a proposed measure.
For example, \citet{Dub75} and \citet{DubS79} have characterized
the classic Shapely-Shubik index ($SS$) and Penrose-Banzhaf measure ($PB$)
as uniquely satisfying a distinct set of axioms, respectively,
but, as critics have noted, several of the axioms lack proper justification:%
\footnote{For the introduction of these measures, see \citet{ShaS54} for $SS$ and \citet{Pen46,Ban65,Ban66} for $PB$.
See also \citet{FelM98, FelM04, LarV08}.}
the {\em additivity} (or {\em transfer}) postulate that both share is unmotivated,
and the postulates distinguishing the two ({\em efficiency} for $SS$ and {\em total power} for $PB$)
are either unnecessary or ad hoc (\citealp[pp.~292-96]{Str82}, \citealp[pp. 194-5]{FelM98}, \citealp{LarV01}).

The second, {\em two-pronged} approach is more modest and involves combining two prongs of justification.
The first prong is to motivate a proposed measure on conceptual grounds,
showing the sense in which it captures the core features of the concept of voting power.
With this conceptual justification in place,
the second prong then requires showing that the measure satisfies
a set of postulates we should expect any reasonable measure of voting power to satisfy.
For the more modest approach, both prongs of justification are necessary:
on the one hand, because more than one measure may satisfy the set of reasonable axioms,
we must turn to conceptual justification to adjudicate between competing proposals;
on the other hand, any violations of reasonable postulates count against a measure
regardless of how intuitive an interpretation can be provided for its conceptual meaning.
Thus, for this two-pronged approach, the satisfaction of reasonable postulates serves, not to pick out a uniquely reasonable measure,
but to rule out unreasonable measures.

The first prong of justification has been typically carried out in {\em probabilistic} terms.
For example, the {\em a priori} Penrose-Banzhaf measure equates a player's voting power, in a given voting structure, 
with the proportion of logically possible {\em divisions} or complete vote configurations 
in which the player is (fully) {\em decisive} for the division outcome,
i.e., in which the player has an alternative voting strategy such that, if it were to choose that alternative 
instead, the outcome would be different (holding all other votes constant).
The standard interpretation is that the a priori $PB$ measure represents the probability a player will be decisive under the 
assumptions of {\em equiprobable voting} (the probability a player votes for an alternative is equal to the probability it 
votes for any other) and {\em voting independence} (votes are not correlated),
which together imply
{\em equiprobable divisions} \citep[pp. 37-38]{FelM98}.
(The equiprobable-divisions assumption is supposed to model a priori power because the latter refers to voting power solely
in virtue of the formal voting structure, abstracted from the distribution of preferences.)
The classic {\em a priori} $PB$ measure is a special case of a generalized measure that weights a player's decisiveness in 
each division by that division's probability;
the generalized measure therefore represents the probability a player will
be decisive given some probability distribution for the divisions.
If each division is weighted by its actual ex ante probability --
given the actual distribution of players' preferences and the potential effects of strategic considerations on voting behaviour --
then the generalized measure yields a measure of so-called {\em a posteriori} voting power.%
\footnote{On the distinction between {\em a priori} and {\em a posteriori} voting power, see \citet{FelM03, FelM04}.}

Thus the first prong of justification for $PB$ relies on showing the intuitive plausibility of equating voting power with the probability of decisiveness,
by arguing in favour of equating the notion of having efficaciously exercised power to effect an outcome with that of being decisive for it.
Similarly, $SS$ has been interpreted in probabilistic terms as the probability
a player will be decisive if players share a common standard by which they judge the desirability of alternatives,
which can be formalized as the probability of decisiveness given a probability distribution of divisions
resulting from ``homogeneous'' voting behaviour, that is,
if the probability any player votes for some arbitrary alternative is the same for all players
and selected from a uniform distribution on [0,1] \citep{Owe75,Str77,Lee90,LarV05b}.
On this probabalistic interpretation, $SS$ is not a measure of a priori voting power,
but of a posteriori voting power under such a homogeneous probability distribution assumption.%
\footnote{As an a priori index, by contrast, $SS$ has been interpreted as measuring the relative {\em value}
of a player's a priori voting power, and therefore, for example, as a bribe index \citep{Mor02},
or as measuring the player's expected payoff assuming a cooperative game with transferable utility \citep{FelM98}.}

However, measures of voting power based exclusively on the ex ante probability of decisiveness
suffer from a crucial conceptual flaw.
The motivation for basing a measure of voting power on this notion is that decisiveness
is supposed to formalize the idea of a player {\em making a difference} to the outcome.
To equate a player's voting power with the player's ex ante probability of being decisive
is to assume that if any particular division were hypothetically to occur,
then the player would have efficaciously exercised power to help produce the outcome ex post if and only if
that player would have been decisive or necessary for the outcome.
Yet this assumption is false: sometimes an actor has efficaciously exercised its power to effect an outcome ex post, and,
through the exercise of that power, made a causal contribution to the outcome,
even though the actor's contribution was not decisive to it.%
\footnote{We presuppose a notion of active power as a conditional-dispositional property \citep{Mor02}.}

This is the case, for example, for {\em causally overdetermined} outcomes.
Consider a three-player vote under majority rule. In a unanimous 3-0 {\sc yes}-vote, no single individual 
player is (fully) decisive for the outcome:
for any player, even if that player had voted {\sc no}, the {\sc yes}-outcome would have remained intact.
Yet it would be a mistake to conclude that, because no single player has ``made a difference'' to the outcome,
in the sense of being decisive,
none has, by exercising its voting power, helped to cause it.
The notion of exercising power to effect an outcome is broader than the notion of making a difference
(or being decisive).
More specifically, reducing voting power to the ex ante probability of being decisive fails to take into account
players' {\em partial} causal efficacy in producing outcomes ex post.
This failure is why $PB$ interprets each individual player, in the unanimous 3-0 division,
as not having efficaciously exercised any voting power ex post at all --
even though in fact each player causally contributes and hence is partially efficacious in realizing it.
Decisiveness measures of voting power falter on the first, conceptual prong of justification.%
\footnote {For more extensive defence of this point, see \citet{Abi22}.
On partial causation, see also \citet{Wri85,Wri88,McD95,Ram97,Hit01,HalP05,Hal07,BraH09}.
It might be objected that overdetermined outcomes may be caused by the mereological sum of individuals, rather than by any of the individuals in particular \citep[pp. 181-2]{Lew86, Bar02}.
But as \citet{Sch03} has argued, it is wholly implausible to attribute emergent causal properties to a collective none of whose individual members plays a causal role.}

In this paper, we design a {\em Recursive Measure} ($RM$) of voting power that remedies this shortcoming,
by taking into account partial efficacy or degrees of causal efficacy.
To ask whether a player would have been decisive or {\em fully efficacious} if various divisions were to have occurred
is to ask a set of hypothetical questions about what would counterfactually be the case if a given vote configuration were to arise.
Similarly, to ask whether a player would have been {\em partially efficacious} within a particular division,
we pose a further series of {\em nested} hypothetical questions counterfactualizing about that division itself.
For example, for any division whose outcome is causally overdetermined, we ask:
Would the player have been decisive if a division with the same outcome had occurred
that was identical except that one player who voted in favour of the outcome were to vote against it?
And in what proportion of such doubly counterfactualized, outcome-preserving divisions would the player be decisive?
In the unanimous 3-0 division under majority rule,
there are three such doubly counterfactualized divisions, each of which preserves the {\sc yes} outcome by 2-1.
And each {\sc yes}-voter in the 3-0 division would be decisive in two of these three hypothetical divisions.
This yields a measure of the player's partial efficacy in the unanimous division (namely, $\frac{2}{3}$).
And if the doubly counterfactualized divisions are themselves causally overdetermined,
then we must of course recursively iterate the calculation for them.

This is how $RM$ is constructed, which is why we call it the Recursive Measure.
A full conceptual justification of such a measure -- i.e., the first prong of justification on the more modest approach --
is given in \citet{Abi22}.
The key to this justification lies in the fact that $RM$ does not reduce the efficacious exercise of voting power to being decisive;
the measure is grounded, rather, in the broader notion of causal efficacy.
$RM$ represents, not the {\em probability} a player will be decisive for the division outcome (the probability the player 
will be {\em fully causally efficacious} in bringing it about) but, rather, the player's {\em expected efficacy}, that is, the probability the 
player will make a causal contribution to the outcome weighted by the degree of causal efficacy.
Whereas decisiveness measures such as $PB$ solely track full efficacy, $RM$ tracks partial efficacy as well.

Yet however strong the conceptual justification for such a measure in general,
we also need to justify its {\em specific} construction or formulation.
Moreover, no matter how intuitively plausible, and no matter how justified its conceptual foundations,
$RM$ would not be a viable measure of voting power unless it also satisfied a number of postulates that arguably any reasonable measure ought to satisfy.
The more modest approach accordingly requires supplementing the first,
conceptual prong of justification with the second, postulate-satisfaction prong.
Our task in this paper is therefore two-fold:
first, to justify the specific formulation we give to the Recursive Measure;
and second, to furnish the second prong of justification given this formulation.
In particular, we take it that any reasonable measure of a priori voting power should satisfy, for simple voting games,
the {\em iso-invariance}, {\em dummy}, {\em dominance}, {\em donation}, {\em minimum-power bloc}, and {\em quarrel} postulates.
We here explain the intuitive justification for each of these voting-power postulates,
and then prove that $RM$ satisfies them for a priori power in simple voting games.
Moreover, we prove these by introducing a new way of representing voting games using a division lattice.

\section{The Voting Model}\label{sec:model}

In this section we present voting games and, in particular, the class of simple voting games
ubiquitous in the literature. We then explain how voting games can be represented by the division lattice.
This lattice representation will be used in Section~\ref{sec:recursive} to design a recursive measure
of voting power that incorporates partial causal efficacy.

\subsection{Simple Voting Games}\label{sec:simple}
Let $[n]=\{1,2,\dots, n\}$ be a nonempty, finite set of players
with two strategies, voting {\sc yes} or voting {\sc no}, 
and let $\mathcal{O}$=\{{\sc yes}, {\sc no}\} be the set of alternative outcomes.
A division $\mathbb{S}=(S, \bar{S})$ of the set $[n]$ is an ordered partition of players where the first element in the 
ordered pair is the set of {\sc yes}-voters 
and the second element is the set of {\sc no}-voters in $\mathbb{S}$. Thus, for $\mathbb{S}=(S, \bar{S})$, the subset 
$S\subseteq [n]$ comprises the set of {\sc yes}-voters 
and the subset $\bar{S}= [n]\setminus S$ comprises the set of {\sc no}-voters.
Note the convention of representing a bipartitoned division by its first element in blackboard bold.

Let $\mathcal{D}$ be the set of all logically possible divisions $\mathbb{S}$ of $[n]$.
A {\em binary voting game}, in which each player has two possible strategies,
is a function $\mathcal{G}(\mathbb{S})$ mapping the set of all possible divisions $\mathcal{D}$ to two outcomes in $\mathcal{O}$.
A {\em monotonic} binary voting game is one satisfying the condition:\\
\indent (i) {\tt Monotonicity.} If $\mathcal{G}$($\mathbb{S}$)={\sc yes} and $S\subseteq T$, then $\mathcal{G}(\mathbb{T})$={\sc yes}.\\
{\em Monotonicity} states that if a division outcome is {\sc yes}, then the outcome of any division in which at least 
the same players vote {\sc yes} will also be {\sc yes}.
Hence, monotonicity states that if a set of players could ensure a {\sc yes}-outcome by each voting {\sc yes}, 
then any superset of those players could do so as well.
A {\em simple voting game} is a monotonic binary voting game that satisfies the additional condition:\\
\indent (ii) {\tt Non-Triviality.} $\exists \mathbb{S}$ $\mathcal |\ {G}$($\mathbb{S}$)={\sc yes} 
and $\exists \mathbb{S}$ $\mathcal |\ {G}$($\mathbb{S}$)={\sc no}.\\
{\em Non-Triviality} states that not all divisions yield the same outcome,
i.e., there is at least one division whose outcome is {\sc yes} and at least one whose outcome is {\sc no}.
Together, monotonicity and non-triviality ensure that simple voting games also have the
property that if everyone votes {\sc no}, the outcome is {\sc no}, and if everyone votes {\sc yes}, the 
outcome is {\sc yes}.\\
\indent (iii) {\tt Unanimity.} $\mathcal{G}((\emptyset, [n]))$={\sc no} and $\mathcal{G}(([n], \emptyset))$={\sc yes}.

\noindent We remark that {\em unanimity} immediately implies non-triviality. Thus conditions (i) and (iii) also characterize the class of simple voting games.

Call any player whose vote corresponds to the division outcome a {\em successful} player.
Let $\mathcal{W}$ be the collection of all sets of players $S$ such that $\mathcal{G}$($\mathbb{S}$)={\sc yes}
(if each member of $S$ were to vote {\sc yes}, they would be successful {\sc yes}-voters).
We call this the collection of {\em {\sc yes}-successful subsets} of $[n]$, also commonly called {\em winning coalitions}.
We can now alternatively characterize conditions (i)-(iii) as:\\
\indent (i) {\tt Monotonicity.} If $S\in \mathcal{W}$ and $S\subseteq T$ then $T\in \mathcal{W}$.\\
\indent (ii) {\tt Non-Triviality.} $\exists S |$ $S \in \mathcal{W}$ and $\exists S |$ $S \notin \mathcal{W}$\\
\indent (iii) {\tt Unanimity.} $[n]\in \mathcal{W}$ and $\emptyset \notin \mathcal{W}$.\\
In the discussion and proofs that follow, it should be understood that, as is standard in the voting-power literature,
we are discussing simple voting games so defined.

\subsection{The Division Lattice}\label{sec:lattice}

The divisions of a voting game can be plotted on a lattice, called the {\em division lattice} $\mathcal{L}=(\mathcal{D}, \succeq)$.
There is an element $\mathbb{S}$ in the lattice for each ordered division $\mathbb{S}=(S, \bar{S})$.
The elements of the lattice are ordered by comparing the sets of players who vote {\sc yes} in each division.
Specifically, for $\mathbb{T}=(T, \bar{T})$, we have $\mathbb{S}\succ \mathbb{T}$ if and only if $T\subset S$;
that is, the {\sc yes}-voters in $\mathbb{T}$ are a strict subset of those in $\mathbb{S}$. 
This implies that the supremum (top element) of the lattice is the division $\mathbbm{[n]}=([n], \emptyset)$ where every player votes {\sc yes}. 
Similarly, the infimum (bottom element) of the lattice is the division $(\emptyset, [n])$ where every player votes {\sc no}. 
We shade an element of the lattice grey if the division yields a {\sc yes}-outcome, 
and leave it white if it yields a {\sc no}-outcome. 

Consider the weighted voting game $\mathcal{G}=\{8; 5,4,3,2\}$.
(A {\em weighted voting game} is one in which each player's vote has a fixed weight \citep[pp. 31-32]{FelM98}.)
The number prior to the semicolon is the quota required for a {\sc yes}-outcome;
the numbers afterwards are the weights of each player's vote.
The division lattice for this game is shown in Figure~\ref{fig:division-lattice} where
each logically possible division $\mathbb{S}$ is labelled by its {\sc yes}-voters $S$.

\begin{figure}[h!]
\begin{center}
\begin{tikzpicture}[scale=0.4]
\draw (0,0) circle [radius=1]; \node at (0,0) {{\footnotesize $\emptyset$}};

\draw [thick, -] (0,1) -- (1.5,2); \draw [thick, -] (0,1) -- (4.5,2); \draw [thick, -] (0,1) -- (-1.5,2); \draw [thick, -] (0,1) -- (-4.5,2);

\draw (1.5,3) circle [radius=1]; \node at (1.5,3) {{\footnotesize $3$}};
\draw (4.5,3) circle [radius=1]; \node at (4.5,3) {{\footnotesize $4$}};
\draw (-1.5,3) circle [radius=1]; \node at (-1.5,3) {{\footnotesize $2$}};
\draw (-4.5,3) circle [radius=1]; \node at (-4.5,3) {{\footnotesize $1$}};

\draw [thick, -] (-4.5,4)-- (-4.75,5); \draw [thick, -] (-4.5,4)-- (-7.75,5); \draw [thick, -] (-4.5,4)-- (-1.75,5); 
\draw [thick, -] (-1.5,4)-- (-7.25,5); \draw [thick, -] (-1.5,4)-- (1.25,5); \draw [thick, -] (-1.5,4)-- (4.25,5); 
\draw [thick, -] (1.5,4)-- (-4.5,5); \draw [thick, -] (1.5,4)-- (1.5,5); \draw [thick, -] (1.5,4)-- (7.25,5); 
\draw [thick, -] (4.5,4)-- (-1.5,5); \draw [thick, -] (4.5,4)-- (4.5,5); \draw [thick, -] (4.5,4)-- (7.75,5); 

\draw (1.5,6) circle [radius=1]; \node at (1.5,6) {{\footnotesize $23$}};
\draw (4.5,6) circle [radius=1]; \node at (4.5,6) {{\footnotesize $24$}};
\draw (7.5,6) circle [radius=1]; \node at (7.5,6) {{\footnotesize $34$}};
\draw [fill=gray!20] (-4.5,6) circle [radius=1]; \node at (-4.5,6) {{\footnotesize $13$}};
\draw (-1.5,6) circle [radius=1]; \node at (-1.5,6) {\footnotesize {$14$}};
\draw [fill=gray!20] (-7.5,6) circle [radius=1]; \node at (-7.5,6) {{\footnotesize $12$}};

\draw [thick, -] (1.75,8) -- (7.5,7);\draw [thick, -] (1.5,8) -- (-1.5,7);\draw [thick, -] (1.25,8) -- (-4.5,7);
\draw [thick, -] (4.75,8) -- (7.5,7);\draw [thick, -] (4.5,8) -- (4.5,7);\draw [thick, -] (4.25,8) -- (1.5,7);
\draw [thick, -] (-4.75,8) -- (-7.5,7);\draw [thick, -] (-4.25,8) -- (1.5,7);\draw [thick, -] (-4.5,8) -- (-4.5,7);
\draw [thick, -] (-1.75,8) -- (-7.5,7);\draw [thick, -] (-1.5,8) -- (-1.5,7);\draw [thick, -] (-1.25,8) -- (4.5,7);

\draw  [fill=gray!20] (1.5,9) circle [radius=1]; \node at (1.5,9) {{\footnotesize $134$}};
\draw  [fill=gray!20] (4.5,9) circle [radius=1]; \node at (4.5,9) {{\footnotesize $234$}};
\draw  [fill=gray!20] (-1.5,9) circle [radius=1]; \node at (-1.5,9) {{\footnotesize $124$}};
\draw  [fill=gray!20] (-4.5,9) circle [radius=1]; \node at (-4.5,9) {{\footnotesize  $123$}};

\draw [thick, -] (0.2,11) -- (1.5,10); \draw [thick, -] (0.6,11) -- (4.5,10); \draw [thick, -] (-0.2,11) -- (-1.5,10); \draw [thick, -] (-0.6,11) -- (-4.5,10);

\draw  [fill=gray!20] (0,12) circle [radius=1]; \node at (0,12) {{\footnotesize $1234$}};
\end{tikzpicture}
\end{center}
\caption{The Division Lattice $\mathcal{L}$ and its {\sc yes}-outcomes and {\sc no}-outcomes.}\label{fig:division-lattice}
\end{figure}
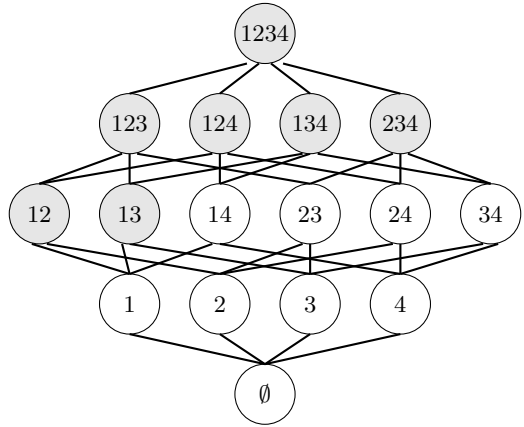

Call a division $\mathbb{S}$ {\em winning} if its outcome is {\sc yes}, {\em losing} if {\sc no}.
The division lattice $\mathcal{L}$ induces two posets, one consisting of the winning {\sc yes}-divisions 
and another of the losing {\sc no}-divisions.
These are illustrated in Figure~\ref{fig:two-posets} for $\{8; 5,4,3,2\}$.

\begin{figure}[h!]
\begin{minipage}{6cm}
\begin{center}
\begin{tikzpicture}[scale=0.4]
\draw [fill=gray!20] (-4.5,6) circle [radius=1]; \node at (-4.5,6) {{\footnotesize $13$}};
\draw [fill=gray!20] (-7.5,6) circle [radius=1]; \node at (-7.5,6) {{\footnotesize $12$}};

\draw [thick, -] (1.25,8) -- (-4.5,7);
\draw [thick, -] (-4.75,8) -- (-7.5,7);\draw [thick, -] (-4.5,8) -- (-4.5,7);
\draw [thick, -] (-1.75,8) -- (-7.5,7);

\draw  [fill=gray!20] (1.5,9) circle [radius=1]; \node at (1.5,9) {{\footnotesize $134$}};
\draw  [fill=gray!20] (4.5,9) circle [radius=1]; \node at (4.5,9) {{\footnotesize $234$}};
\draw  [fill=gray!20] (-1.5,9) circle [radius=1]; \node at (-1.5,9) {{\footnotesize $124$}};
\draw  [fill=gray!20] (-4.5,9) circle [radius=1]; \node at (-4.5,9) {{\footnotesize $123$}};

\draw [thick, -] (0.2,11) -- (1.5,10); \draw [thick, -] (0.6,11) -- (4.5,10); 
\draw [thick, -] (-0.2,11) -- (-1.5,10); \draw [thick, -] (-0.6,11) -- (-4.5,10);

\draw  [fill=gray!20] (0,12) circle [radius=1]; \node at (0,12) {{\footnotesize $1234$}};
\end{tikzpicture}
\end{center}
\end{minipage}
\qquad \quad
\begin{minipage}{6cm}
  \begin{center}
  \begin{tikzpicture}[scale=0.4]
\draw (0,0) circle [radius=1]; \node at (0,0) {{\footnotesize $\emptyset$}};

\draw [thick, -] (0.15,1) -- (1.5,2); \draw [thick, -] (0.5,1) -- (4.5,2); \draw [thick, -] (-0.15,1) -- (-1.5,2); \draw [thick, -] (-.5,1) -- (-4.5,2);

\draw (1.5,3) circle [radius=1]; \node at (1.5,3) {{\footnotesize $3$}};
\draw (4.5,3) circle [radius=1]; \node at (4.5,3) {{\footnotesize $4$}};
\draw (-1.5,3) circle [radius=1]; \node at (-1.5,3) {{\footnotesize $2$}};
\draw (-4.5,3) circle [radius=1]; \node at (-4.5,3) {{\footnotesize $1$}};

 \draw [thick, -] (-4.5,4)-- (-1.75,5); 
 \draw [thick, -] (-1.5,4)-- (1.25,5);  \draw [thick, -] (-1.15,4)-- (4.25,5); 
 \draw [thick, -] (1.25,4)-- (1.5,5); \draw [thick, -] (1.75,4)-- (7.25,5); 
\draw [thick, -] (4.25,4)-- (-1.5,5); \draw [thick, -] (4.5,4)-- (4.5,5); \draw [thick, -] (4.75,4)-- (7.75,5); 

\draw (1.5,6) circle [radius=1]; \node at (1.5,6) {{\footnotesize $23$}};
\draw (4.5,6) circle [radius=1]; \node at (4.5,6) {{\footnotesize $24$}};
\draw (7.5,6) circle [radius=1]; \node at (7.5,6) {{\footnotesize $34$}};
\draw (-1.5,6) circle [radius=1]; \node at (-1.5,6) {{\footnotesize $14$}};
 \end{tikzpicture}
\end{center}
\end{minipage}
\caption{Two Induced Posets.}\label{fig:two-posets}
\end{figure}
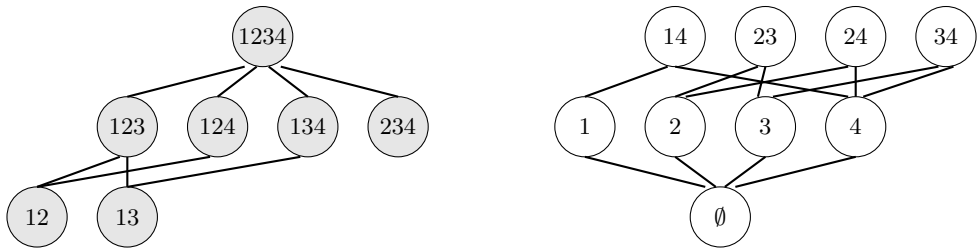

It is useful to invert the poset of losing divisions so that it has a supremum $(\emptyset, [n])$.
Thus we obtain two posets with supremums, called the {\sc yes}-poset and the {\sc no}-poset, respectively. 
See Figure~\ref{fig:yes-no-posets}.

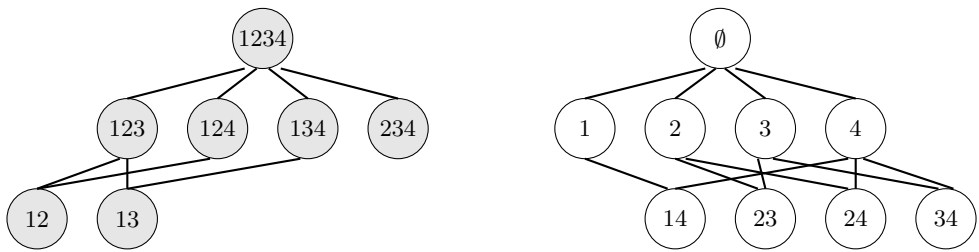
\begin{figure}[h!]
\begin{minipage}{6cm}
\begin{center}
\begin{tikzpicture}[scale=0.4]
\draw [fill=gray!20] (-4.5,6) circle [radius=1]; \node at (-4.5,6) {{\footnotesize $13$}};
\draw [fill=gray!20] (-7.5,6) circle [radius=1]; \node at (-7.5,6) {{\footnotesize $12$}};

\draw [thick, -] (1.25,8) -- (-4.5,7);
\draw [thick, -] (-4.75,8) -- (-7.5,7);\draw [thick, -] (-4.5,8) -- (-4.5,7);
\draw [thick, -] (-1.75,8) -- (-7.5,7);

\draw  [fill=gray!20] (1.5,9) circle [radius=1]; \node at (1.5,9) {{\footnotesize $134$}};
\draw  [fill=gray!20] (4.5,9) circle [radius=1]; \node at (4.5,9) {{\footnotesize $234$}};
\draw  [fill=gray!20] (-1.5,9) circle [radius=1]; \node at (-1.5,9) {{\footnotesize $124$}};
\draw  [fill=gray!20] (-4.5,9) circle [radius=1]; \node at (-4.5,9) {{\footnotesize $123$}};

\draw [thick, -] (0.2,11) -- (1.5,10); \draw [thick, -] (0.6,11) -- (4.5,10); 
\draw [thick, -] (-0.2,11) -- (-1.5,10); \draw [thick, -] (-0.6,11) -- (-4.5,10);

\draw  [fill=gray!20] (0,12) circle [radius=1]; \node at (0,12) {{\footnotesize $1234$}};
\end{tikzpicture}
\end{center}
\end{minipage}
\qquad \quad
\begin{minipage}{6cm}
  \begin{center}
  \begin{tikzpicture}[scale=0.4]
\draw (0,0) circle [radius=1]; \node at (0,0) {{\footnotesize $\emptyset$}};

\draw [thick, -] (0.15,-1) -- (1.5,-2); \draw [thick, -] (0.5,-1) -- (4.5,-2); \draw [thick, -] (-0.15,-1) -- (-1.5,-2); \draw [thick, -] (-.5,-1) -- (-4.5,-2);

\draw (1.5,-3) circle [radius=1]; \node at (1.5,-3) {{\footnotesize $3$}};
\draw (4.5,-3) circle [radius=1]; \node at (4.5,-3) {{\footnotesize $4$}};
\draw (-1.5,-3) circle [radius=1]; \node at (-1.5,-3) {{\footnotesize $2$}};
\draw (-4.5,-3) circle [radius=1]; \node at (-4.5,-3) {{\footnotesize $1$}};

 \draw [thick, -] (-4.5,-4)-- (-1.75,-5); 
 \draw [thick, -] (-1.5,-4)-- (1.25,-5);  \draw [thick, -] (-1.15,-4)-- (4.25,-5); 
 \draw [thick, -] (1.25,-4)-- (1.5,-5); \draw [thick, -] (1.75,-4)-- (7.25,-5); 
\draw [thick, -] (4.25,-4)-- (-1.5,-5); \draw [thick, -] (4.5,-4)-- (4.5,-5); \draw [thick, -] (4.75,-4)-- (7.75,-5); 

\draw (1.5,-6) circle [radius=1]; \node at (1.5,-6) {{\footnotesize $23$}};
\draw (4.5,-6) circle [radius=1]; \node at (4.5,-6) {{\footnotesize $24$}};
\draw (7.5,-6) circle [radius=1]; \node at (7.5,-6) {{\footnotesize $34$}};
\draw (-1.5,-6) circle [radius=1]; \node at (-1.5,-6) {{\footnotesize $14$}};
 \end{tikzpicture}
\end{center}
\end{minipage}
\caption{The {\sc yes}-poset and the {\sc no}-poset.}\label{fig:yes-no-posets}
\end{figure}

%

In the {\sc yes}-poset, we say that $\mathbb{T}$ is a {\em loyal child} of $\mathbb{S}$ (and $\mathbb{S}$ is a {\em loyal parent} of $\mathbb{T}$)
if and only if $S=T\cup \{j\}$. That is, $\mathbb{T}$ is identical to $\mathbb{S}$ except that exactly one less player votes {\sc yes} in $\mathbb{T}$ than in $\mathbb{S}$.
The nomenclature {\em loyal} refers to the fact that $\mathbb{S}$ and $\mathbb{T}$ have the same outcome.
Symmetrically, in the {\sc no}-poset, we say that $\mathbb{T}$ is a {\em loyal child} of $\mathbb{S}$ (and $\mathbb{S}$ is a {\em loyal parent} of $\mathbb{T}$)
if and only if $S=T\setminus \{j\}$. That is, $\mathbb{T}$ is identical to $\mathbb{S}$ except that exactly one less player votes {\sc no} in $\mathbb{T}$ than 
in $\mathbb{S}$.
Moreover, we call a division's {\em loyal descendants} those divisions that are its loyal children, their loyal children, and so on.

We maintain the same terminology when describing the entire division lattice $\mathcal{L}$,
denoting by $LC(\mathbb{S})$ the set of loyal children of $\mathbb{S}$ in $\mathcal{L}$.
Notice that this implies that if $\mathbb{S}$ is a winning {\sc yes}-division then its loyal children, if it has any, lie immediately {\em beneath} it in $\mathcal{L}$.
By contrast, since the {\sc no}-poset was inverted,
if $\mathbb{S}$ is a losing {\sc no}-division then its loyal children, if it has any,  lie immediately {\em above} it in the division lattice $\mathcal{L}$.



\section{The Recursive Measure of Voting Power}\label{sec:recursive}

Our goal is to construct a measure of voting power that incorporates partial causal efficacy,
which, we suggested, requires a recursive measure.
How exactly should such a measure be defined?
Our formulation, presented in Section~\ref{sec:RM-definition},
is motivated by the division lattice representation of voting games, in particular,
via the loyal children concept proffered by the {\sc yes}-poset and {\sc no}-poset.
We provide an example illustrating how the measure is calculated using a weighted voting game in Section~\ref{sec:example}
and an interpretation of the measure, in terms of random walks on the {\sc yes}-poset and {\sc no}-poset, in Section~\ref{sec:random-walks}.
The strength of our proposed measure will be demonstrated in Sections~\ref{sec:minimal-postulates} to \ref{sec:quarrel-postulate}
where we show it satisfies a series of reasonable voting postulates.

\subsection{A Recursive Formulation}\label{sec:RM-definition}

We define a {\em measure of voting power} for simple voting games as a function $\Psi$ that  assigns to each player $i$ a nonnegative real number $\Psi_i \geq 0$ and that satisfies two sets of basic adequacy postulates:
the {\em iso-invariance} postulate, according to which the a priori voting power of any player remains the same between two isomorphic games;
and the {\em dummy} postulates, according to which a player has zero a priori voting power if and only if it is a dummy (i.e., not decisive in any division),
and the addition of a dummy to a voting structure leaves other players' a priori voting power unchanged \citep[p. 236]{FelM98} (see section \ref{sec:minimal-postulates} below).
(When we refer to $\Psi$'s measure of a priori voting power, we shall represent it using the lower case $\psi$.)
We define an {\em efficacy measure} of voting power (of which decisiveness measures such as $PB$ are a species)
as a measure based strictly on the causal efficacy of the player in each of the voting game's divisions, weighted by the significance of each division.

An efficacy measure of voting power $\Psi$ therefore assigns to each player $i$ a value
\begin{align*}
\Psi_i &= \sum_{\mathbb{S}\in \mathcal{D}} \alpha_i(\mathbb{S})\cdot \gamma(\mathbb{S})
\end{align*}
where $\alpha_i(\mathbb{S})$ is the {\em division efficacy score} of player $i$ in division $\mathbb{S}$
and $\gamma(\mathbb{S})$ is the {\em division weight} assigned to each division $\mathbb{S}\in \mathcal{D}$.
The defining characteristic of a given measure of voting power is therefore its specification
of a player's division efficacy score for each division and each division's weight.

For example, the classic $PB$ measure models a priori voting power by assigning each division an equal weight,
such that $\gamma^{PB}(\mathbb{S})=\frac{1}{|\mathcal{D}|}$, which, for binary voting games, amounts to $\frac{1}{2^n}$.
Here $\gamma(\mathbb{S})$ is typically interpreted as representing the ex ante probability $\mathbb{P}(\mathbb{S})$ of each division $\mathbb{S}\in \mathcal{D}$,
where for a priori voting power each division is assumed to be equiprobable.
(For a posteriori voting power, each division's weight (probability) depends on the distribution voter preferences.)
The classic $PB$ measure defines the division efficacy score, in turn, {\em non-recursively} as
$$
\alpha^{PB}_i(\mathbb{S}) =
\begin{cases}
1& \mathrm{if\ } i \mathrm{\ is\ decisive\ in\ } \mathbb{S}\\
0 & \mathrm{otherwise}
\end{cases}
$$
A player $i$ is {\sc yes}{\em -decisive} in division $\mathbb{S}$ if and only if $i\in S\in \mathcal{W}$ but $S\setminus\{i\}\notin \mathcal{W}$;
is {\sc no}{\em -decisive} if and only if $i\notin S\notin \mathcal{W}$ but $S\cup\{i\}\in \mathcal{W}$;
and is {\em decisive} if and only if it is either {\sc yes}-decisive or {\sc no}-decisive.

Given the conceptual shortcomings inherent to a measure based only on full causal efficacy,
we incorporate partial causal efficacy via a recursive definition of the division efficacy score
(whilst equating, as with $PB$, the division weight to the division probability, $\gamma(\mathbb{S})=\mathbb{P}(\mathbb{S})$).
In particular, the {\em Recursive Measure of Voting Power} ($RM$) is defined as
\begin{align}\label{RM-general}
RM_i &= \sum_{\mathbb{S}\in \mathcal{D}} \alpha_i(\mathbb{S})\cdot \mathbb{P}(\mathbb{S})
\end{align}
where $\alpha_i(\mathbb{S})$, the division efficacy score of player $i$ in division $\mathbb{S}$, is defined recursively as
$$
\alpha_i(\mathbb{S}) =
\begin{cases}
1& \mathrm{if\ } i \mathrm{\ is\ decisive\ in\ } \mathbb{S}\\
0 & \mathrm{if\ }  i \mathrm{\ is \ not\ successful\ in\ } \mathbb{S}\\
\frac{1}{|LC(\mathbb{S})|}\cdot \sum_{\mathbb{\hat{S}}\in LC(\mathbb{S})} \alpha_i(\hat{\mathbb{S}}) & \mathrm{otherwise}
\end{cases}
$$
where $LC(\mathbb{S})$ is the set of loyal children of $\mathbb{S}$ in the division lattice.

Definition~(\ref{RM-general}) gives $RM$ in generalized form.
To measure {\em a priori} voting power according to $RM$, which we represent as $RM'$,
we assume equiprobable divisions and therefore set $\mathbb{P}=\frac{1}{|\mathcal{D}|}$. Thus
\begin{align}\label{RM-apriori}
RM'_i &= \frac{1}{|\mathcal{D}|}\cdot \sum_{\mathbb{S}\in \mathcal{D}} \alpha_i(\mathbb{S})
\end{align}
where, for binary voting games, $\frac{1}{|\mathcal{D}|}=\frac{1}{2^n}$.

To compute the efficacy score $\alpha$  for a player we distinguish between its {\sc yes}{\em -efficacy score}~$\alpha^+$ 
and its {\sc no}{\em -efficacy score}~$\alpha^-$. For $RM$ these are defined recursively as:
$$
\alpha^+_i(\mathbb{S}) =
\begin{cases}
1& \mathrm{if\ } i \mathrm{\ is\ } \textsc{yes}\mathrm{-decisive\ in\ } \mathbb{S}\\
0 & \mathrm{if\ }  \mathbb{S} \mathrm{\ is \ losing\ or\ if\ } i\notin S\\
\frac{1}{|LC(\mathbb{S})|}\cdot \sum_{\hat{\mathbb{S}}\in LC(\mathbb{S})} \alpha^+_i(\hat{\mathbb{S}}) & \mathrm{otherwise}
\end{cases}
$$

$$
\alpha^-_i(\mathbb{S}) =
\begin{cases}
1& \mathrm{if\ } i \mathrm{\ is\ }\textsc{no}\mathrm{-decisive\ in\ } \mathbb{S}\\
0 &  \mathrm{if\ }  \mathbb{S} \mathrm{\ is \ winning\ or\ if\ } i\in S\\
\frac{1}{|LC(\mathbb{S})|}\cdot \sum_{\hat{\mathbb{S}}\in LC(\mathbb{S})} \alpha^-_i(\hat{\mathbb{S}}) & \mathrm{otherwise}
\end{cases}
$$
The efficacy score is then the sum of the {\sc yes}-efficacy and {\sc no}-efficacy scores, 
$\alpha_i(\mathbb{S})=\alpha^+_i(\mathbb{S})+\alpha^-_i(\mathbb{S})$.
Correspondingly, $RM$ can be written as the sum of a {\em Recursive Measure of {\sc Yes}-Voting Power} $RM^+$
and a {\em Recursive Measure of {\sc No}-Voting Power} $RM^-$:
\begin{align*}
RM_i
&= \sum_{\mathbb{S}\in \mathcal{D}} (\alpha^+_i(\mathbb{S})+\alpha^-_i(\mathbb{S}))\cdot \mathbb{P}(\mathbb{S}) \\
&= \sum_{\mathbb{S}\in \mathcal{D}} \alpha^+_i(\mathbb{S})\cdot \mathbb{P}(\mathbb{S}) + \sum_{\mathbb{S}\in \mathcal{D}} \alpha^-_i(\mathbb{S})\cdot \mathbb{P}(\mathbb{S})\\
 &= RM^+_i + RM^-_i 
\end{align*}
Note that decisiveness measures of a priori voting power such as $PB$ can be computed using a shortcut.
Precisely because such measures only consider full decisiveness,
their measure of a priori {\sc yes}- and {\sc no}-voting power will be equal.
This is because, by definition of decisiveness, for each winning division in which a voter is {\sc yes}-decisive,
there is exactly one corresponding losing division in which the voter is {\sc no}-decisive.
It follows that $PB^+=PB^-$ and that $PB$ is therefore equal to
\begin{align*}
PB_i
&= \frac{\sum_{\mathbb{S}\in \mathcal{D}} \alpha^{+ PB}_i(\mathbb{S})}{2^{n-1}}
\end{align*}
This symmetry between a priori {\sc yes}- and {\sc no}-voting power does not hold, however, for efficacy measures in general.
In a division in which a voter is only partially {\sc yes}-efficacious,
in the corresponding division in which all other votes are held constant but in which the player votes {\sc no},
the voter will be zero {\sc no}-efficacious because the player will now be unsuccessful.
This is why we cannot calculate $RM'$ using a shortcut formula analogous to the one typically used for $PB$,
but must do so on the basis of both $RM'^+$ and $RM'^-$.

\subsection{Calculating Efficacy Scores via the Yes-Poset and No-Poset}\label{sec:example}
We can calculate the efficacy scores of player $i$ via the {\sc yes}-Poset and {\sc no}-Poset.
Again we illustrate this using the posets of Figure~\ref{fig:yes-no-posets} for
the weighted game $\{8; 5,4,3,2\}$.
Using these two posets we recursively find the efficacy scores of each player at each node. 

For example, let's calculate the efficacy scores for player 2.
In the {\sc yes}-poset, the divisions $\mathbb{S}$ where player 2 is {\sc yes}-decisive in Figure~\ref{fig:ES-recursive} are striped downwards (from left to right) and its {\sc yes}-efficacy score is $\alpha^+_2(\mathbb{S})=1$.
For any winning division where player 2 votes {\sc no}, its {\sc yes}-efficacy score is $\alpha^+_2(\mathbb{S})=0$.
(Note that the {\sc yes}-efficacy score of each player, including player 2, is zero for every node in the {\sc no}-poset, since they are all losing divisions.)
The values of the remaining nodes for player 2 in the {\sc yes}-poset are then calculated recursively; see Figure~\ref{fig:ES-recursive}.
Similarly, in the {\sc no}-poset, the nodes where player 2 is {\sc no}-decisive are striped upwards and its {\sc no}-efficacy score is $\alpha^-_2(\mathbb{S})=1$.
For any losing division where player 2 votes {\sc yes}, its {\sc no}-efficacy score is $\alpha^-_2(\mathbb{S})=0$.
(Note that the {\sc no}-efficacy score of each player, including player 2, is zero for every node in the {\sc yes}-poset, since they are all winning divisions.)
The values of the remaining nodes for player 2 in the {\sc no}-poset are then calculated recursively; again, see Figure~\ref{fig:ES-recursive}.

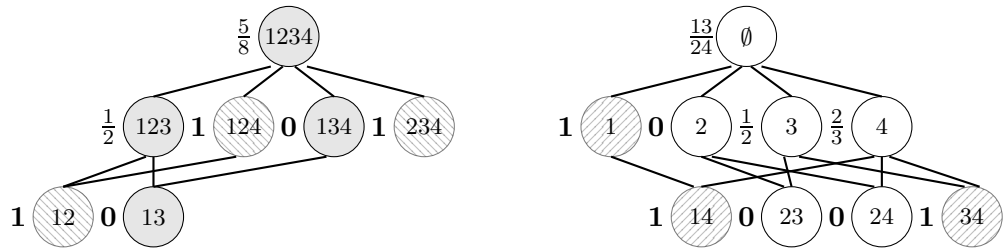
\begin{figure}[h!]
\begin{minipage}{6cm}
\begin{center}
\begin{tikzpicture}[scale=0.4]
\draw [fill=gray!20] (-4.5,6) circle [radius=1]; \node at (-4.5,6) {{\footnotesize $13$}};  \node at (-6,6) {{\bf 0}};
\draw [pattern=north west lines, opacity=.5] (-7.5,6) circle [radius=1]; \node at (-7.5,6) {{\footnotesize $12$}}; \node at (-9,6) {{\bf 1}};

\draw [thick, -] (1.25,8) -- (-4.5,7);
\draw [thick, -] (-4.75,8) -- (-7.5,7);\draw [thick, -] (-4.5,8) -- (-4.5,7);
\draw [thick, -] (-1.75,8) -- (-7.5,7);

\draw  [fill=gray!20] (1.5,9) circle [radius=1]; \node at (1.5,9) {{\footnotesize $134$}}; \node at (0,9) {{\bf 0}};
\draw  [pattern=north west lines, opacity=.5] (4.5,9) circle [radius=1]; \node at (4.5,9) {{\footnotesize $234$}};\node at (3,9) {{\bf 1}};
\draw  [pattern=north west lines, opacity=.5] (-1.5,9) circle [radius=1]; \node at (-1.5,9) {{\footnotesize $124$}};\node at (-3,9) {{\bf 1}};
\draw  [fill=gray!20] (-4.5,9) circle [radius=1]; \node at (-4.5,9) {{\footnotesize $123$}};\node at (-6,9) {{\bf $\frac12$}};

\draw [thick, -] (0.2,11) -- (1.5,10); \draw [thick, -] (0.6,11) -- (4.5,10); \draw [thick, -] (-0.2,11) -- (-1.5,10); \draw [thick, -] (-0.6,11) -- (-4.5,10);

\draw  [fill=gray!20] (0,12) circle [radius=1]; \node at (0,12) {{\footnotesize $1234$}};\node at (-1.5,12) {{\bf $\frac58$}};
\end{tikzpicture}
\end{center}
\end{minipage}
\qquad \quad
\begin{minipage}{6cm}
  \begin{center}
  \begin{tikzpicture}[scale=0.4]
\draw (0,0) circle [radius=1]; \node at (0,0) {{\footnotesize $\emptyset$}};\node at (-1.5,0) {{\bf $\frac{13}{24}$}};

\draw [thick, -] (0.15,-1) -- (1.5,-2); \draw [thick, -] (0.5,-1) -- (4.5,-2); \draw [thick, -] (-0.15,-1) -- (-1.5,-2); \draw [thick, -] (-.5,-1) -- (-4.5,-2);

\draw (1.5,-3) circle [radius=1]; \node at (1.5,-3) {{\footnotesize $3$}}; \node at (0,-3) {{\bf $\frac12$}};
\draw (4.5,-3) circle [radius=1]; \node at (4.5,-3) {{\footnotesize $4$}}; \node at (3,-3) {{\bf $\frac23$}};
\draw (-1.5,-3) circle [radius=1]; \node at (-1.5,-3) {{\footnotesize $2$}};  \node at (-3,-3) {{\bf 0}};
\draw  [pattern=north east lines, opacity=.5]  (-4.5,-3) circle [radius=1]; \node at (-4.5,-3) {{\footnotesize $1$}}; \node at (-6,-3) {{\bf 1}};

 \draw [thick, -] (-4.5,-4)-- (-1.75,-5); 
 \draw [thick, -] (-1.5,-4)-- (1.25,-5);  \draw [thick, -] (-1.15,-4)-- (4.25,-5); 
 \draw [thick, -] (1.25,-4)-- (1.5,-5); \draw [thick, -] (1.75,-4)-- (7.25,-5); 
\draw [thick, -] (4.25,-4)-- (-1.5,-5); \draw [thick, -] (4.5,-4)-- (4.5,-5); \draw [thick, -] (4.75,-4)-- (7.75,-5); 

\draw (1.5,-6) circle [radius=1]; \node at (1.5,-6) {{\footnotesize $23$}}; \node at (0,-6) {{\bf 0}};
\draw (4.5,-6) circle [radius=1]; \node at (4.5,-6) {{\footnotesize $24$}}; \node at (3,-6) {{\bf 0}};
\draw  [pattern=north east lines, opacity=.5] (7.5,-6) circle [radius=1]; \node at (7.5,-6) {{\footnotesize $34$}}; \node at (6,-6) {{\bf 1}};
\draw  [pattern=north east lines, opacity=.5] (-1.5,-6) circle [radius=1]; \node at (-1.5,-6) {{\footnotesize $14$}}; \node at (-3,-6) {{\bf 1}};

 \end{tikzpicture}
\end{center}
\end{minipage}
\caption{Calculating the Efficacy Scores Recursively.}\label{fig:ES-recursive}
\end{figure}

These values can be shown on a single picture using the division lattice $\mathcal{L}$ as in Figure~\ref{fig:ES-values}.

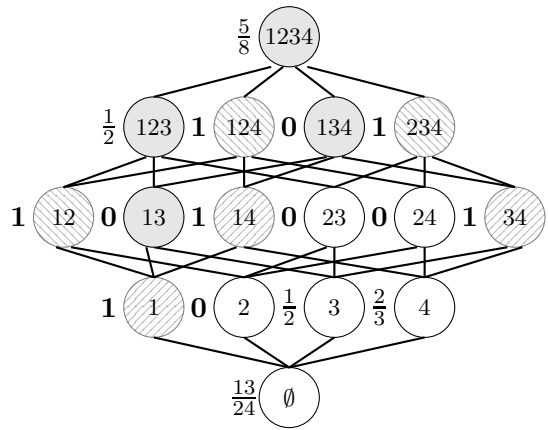
\begin{figure}[h!]
\begin{center}
\begin{tikzpicture}[scale=0.4]
\draw (0,0) circle [radius=1]; \node at (0,0) {{\footnotesize $\emptyset$}}; \node at (-1.5,0) {{\bf $\frac{13}{24}$}};

\draw [thick, -] (0,1) -- (1.5,2); \draw [thick, -] (0,1) -- (4.5,2); \draw [thick, -] (0,1) -- (-1.5,2); \draw [thick, -] (0,1) -- (-4.5,2);

\draw (1.5,3) circle [radius=1]; \node at (1.5,3) {{\footnotesize $3$}}; \node at (0,3) {{\bf $\frac12$}};
\draw (4.5,3) circle [radius=1]; \node at (4.5,3) {{\footnotesize $4$}}; \node at (3,3) {{\bf $\frac23$}};
\draw (-1.5,3) circle [radius=1]; \node at (-1.5,3) {{\footnotesize $2$}};  \node at (-3,3) {{\bf 0}};
\draw  [pattern=north east lines, opacity=.5]  (-4.5,3) circle [radius=1]; \node at (-4.5,3) {{\footnotesize $1$}}; \node at (-6,3) {{\bf 1}};

\draw [thick, -] (-4.5,4)-- (-4.75,5); \draw [thick, -] (-4.5,4)-- (-7.75,5); \draw [thick, -] (-4.5,4)-- (-1.75,5); 
\draw [thick, -] (-1.5,4)-- (-7.25,5); \draw [thick, -] (-1.5,4)-- (1.25,5); \draw [thick, -] (-1.5,4)-- (4.25,5); 
\draw [thick, -] (1.5,4)-- (-4.5,5); \draw [thick, -] (1.5,4)-- (1.5,5); \draw [thick, -] (1.5,4)-- (7.25,5); 
\draw [thick, -] (4.5,4)-- (-1.5,5); \draw [thick, -] (4.5,4)-- (4.5,5); \draw [thick, -] (4.5,4)-- (7.75,5); 

\draw (1.5,6) circle [radius=1]; \node at (1.5,6) {{\footnotesize $23$}};  \node at (0,6) {{\bf 0}};
\draw (4.5,6) circle [radius=1]; \node at (4.5,6) {{\footnotesize $24$}};  \node at (3,6) {{\bf 0}};
\draw  [pattern=north east lines, opacity=.5]  (7.5,6) circle [radius=1]; \node at (7.5,6) {{\footnotesize $34$}}; \node at (6,6) {{\bf 1}};
\draw [fill=gray!20] (-4.5,6) circle [radius=1]; \node at (-4.5,6) {{\footnotesize $13$}}; \node at (-6,6) {{\bf 0}};
\draw  [pattern=north east lines, opacity=.5]  (-1.5,6) circle [radius=1]; \node at (-1.5,6) {{\footnotesize $14$}}; \node at (-3,6) {{\bf 1}};
\draw [pattern=north west lines, opacity=.5] (-7.5,6) circle [radius=1]; \node at (-7.5,6) {{\footnotesize $12$}}; \node at (-9,6) {{\bf 1}};

\draw [thick, -] (1.75,8) -- (7.5,7);\draw [thick, -] (1.5,8) -- (-1.5,7);\draw [thick, -] (1.25,8) -- (-4.5,7);
\draw [thick, -] (4.75,8) -- (7.5,7);\draw [thick, -] (4.5,8) -- (4.5,7);\draw [thick, -] (4.25,8) -- (1.5,7);
\draw [thick, -] (-4.75,8) -- (-7.5,7);\draw [thick, -] (-4.25,8) -- (1.5,7);\draw [thick, -] (-4.5,8) -- (-4.5,7);
\draw [thick, -] (-1.75,8) -- (-7.5,7);\draw [thick, -] (-1.5,8) -- (-1.5,7);\draw [thick, -] (-1.25,8) -- (4.5,7);

\draw  [fill=gray!20] (1.5,9) circle [radius=1]; \node at (1.5,9) {{\footnotesize $134$}}; \node at (0,9) {{\bf 0}};
\draw  [pattern=north west lines, opacity=.5] (4.5,9) circle [radius=1]; \node at (4.5,9) {{\footnotesize $234$}};\node at (3,9) {{\bf 1}};
\draw  [pattern=north west lines, opacity=.5] (-1.5,9) circle [radius=1]; \node at (-1.5,9) {{\footnotesize $124$}};\node at (-3,9) {{\bf 1}};
\draw  [fill=gray!20] (-4.5,9) circle [radius=1]; \node at (-4.5,9) {{\footnotesize $123$}};\node at (-6,9) {{\bf $\frac12$}};

\draw [thick, -] (0.2,11) -- (1.5,10); \draw [thick, -] (0.6,11) -- (4.5,10); \draw [thick, -] (-0.2,11) -- (-1.5,10); \draw [thick, -] (-0.6,11) -- (-4.5,10);

\draw  [fill=gray!20] (0,12) circle [radius=1]; \node at (0,12) {{\footnotesize $1234$}}; \node at (-1.5,12) {{\bf $\frac{5}{8}$}};
\end{tikzpicture}
\caption{The Efficacy Scores}\label{fig:ES-values}
\end{center}
\end{figure}

Given the efficacy scores we may compute the a priori $RM'$ voting power of the second player.
Reading from Figure~\ref{fig:ES-recursive}, we see that $RM'^+_2=\frac{1}{16}(\frac58+\frac12+1+0+1+1+0)=\frac{33}{128}$
and $RM'^-_2=\frac{1}{16}(1+0+0+1+1+0+\frac12+\frac23+\frac{13}{24})=\frac{113}{384}$.
Hence $RM'_2=RM'^+_2+RM'^-_2= \frac{33}{128}+ \frac{113}{384}=\frac{212}{384}=\frac{53}{96}$.

\subsection{Random Walks}\label{sec:random-walks}
A natural interpretation of the $RM$ efficacy scores is given by the concept of random walks in stochastic processes.
Specifically, if $\mathbb{S}$ yields a {\sc yes}-outcome, the {\sc yes}-efficacy score of player $i$ is the probability that a uniform random
walk starting at $\mathbb{S}$ in the {\sc yes}-poset ever reaches a node where $i$ is {\sc yes}-decisive.
Here a uniform random walk means that from the initial node the walk travels next to any loyal child with equal probability, 
and then to any loyal child of those nodes with equal probability, and so on.
The random walk terminates when it reaches a node with no loyal children.
For example, in Figure~\ref{fig:ES-recursive}, at the node represented by $S=\{1,2,3\}$ player 2 has a {\sc yes}-efficacy 
score equal to $\frac12$ because
a random walk starting there has two loyal children, and moves to each with probability $\frac12$, and in each case then terminates 
because neither of these loyal children have any loyal children themselves.
The second player is {\sc yes}-decisive in the loyal child represented by $S=\{2,3\}$ but not in the one represented by $S=\{1,3\}$.
Thus half the random walks starting from the node represented by $S=\{1,2,3\}$ pass through a node in which
player 2 is {\sc yes}-decisive; consequently $\alpha^+_2(\mathbb{S})=\frac12$.

Symmetrically, if $\mathbb{S}$ yields a {\sc no}-outcome,
the {\sc no}-efficacy score of player $i$ is the probability that a uniform random walk starting at $\mathbb{S}$ in the {\sc no}-poset 
ever reaches a node where $i$ is {\sc no}-decisive.
Here a uniform random walk means that from the initial node the walk travels next to any loyal child with equal probability, etc., and 
terminates when it reaches a node with no loyal children.
For example, in Figure~\ref{fig:ES-recursive}, at the node $\mathbb{S}=\emptyset$, the second player has a {\sc no}-efficacy score 
equal to $\frac{13}{24}$.
The reader may verify that a random walk starting from $\mathbb{S}=(\emptyset, [n])$ hits a node where $i$ is {\sc no}-decisive with 
probability $\frac{13}{24}$.

This combinatorial view applies because such random walks naturally encode the
recursive formula's defining efficacy scores. This viewpoint, whilst not required in the proofs that follow,
provides additional insight into the proofs' motivation.

\section{Minimal Adequacy Postulates}\label{sec:minimal-postulates}

In the rest of the paper, we assess the strength of $RM$ by testing it against a set of voting postulates.%
\footnote{For a comprehensive treatment of the postulates most typically deemed to be reasonable to expect measures of voting power to satisfy,
see \citet{FelM98, LarV05a}.}
We begin with the basic {\em adequacy postulates},
namely the {\em iso-invariance postulate} and {\em dummy postulates},
which any reasonable measure of voting power ought to satisfy \citep{FelM98},
and whose satisfaction we embedded in the definition of a measure of voting power.
We shall now prove that $RM$ satisfies these postulates.

\subsection{The Iso-Invariance Postulate}
Two voting games $\mathcal{G}$ and $\hat{\mathcal{G}}$ are {\em isomorphic} if each division in the one maps in a 
one-to-one correspondence onto an identical division with the {\bf same} outcome in the other.
In particular, the {\sc yes}-successful sets $\mathcal{W}$ and $\hat{\mathcal{W}}$ are identical after relabelling the players' names.
The iso-invariance postulate requires that any player's a priori voting power in two isomorphic simple voting games be identical.
Specifically, a measure of voting power $\Psi$,
according to which $\Psi_i$ and $\hat{\Psi}_i$ are player $i$'s voting power in $\mathcal{G}$ and $\hat{\mathcal{G}}$, respectively,
satisfies the {\em iso-invariance postulate} if:
\begin{enumerate}
 \item [{\sc (iso)}] For iso-invariant voting games $\mathcal{G}$ and $\hat{\mathcal{G}}$, 
 we have $\psi_i=\hat{\psi}_i$ for any player $i$.
\end{enumerate}
Evidently, a violation of the iso-invariance postulate would be a critical defect:
iso-invariance merely expresses the requirement that a priori voting power depend on nothing
but the structure of the game itself and the position of each player in that structure.

\begin{theorem}\label{thm:iso-postulate}
RM satisfies the iso-invariance postulate.
\end{theorem}

Proofs for this and all subsequent results are presented in the Appendix.
\subsection{The Dummy Postulates}
We say that a player $d$ is a {\em dummy voter} if it is never decisive in any logically possible division.
That is, a division $(S\cup \{d\}, \bar{S}\setminus \{d\})$  yields a {\sc yes}-outcome {\em if and only if} 
the division $(S, \bar{S})$ yields a {\sc yes}-outcome.
(Again adopting the convention of representing a bipartitioned division by its first element in blackboard bold,
we represent the former division as  $\mathbb{S}\cup \mathbbm{d}$ and the latter as $\mathbb{S}$.)

The dummy postulates require, in a simple voting game, that {\em all} dummies have zero a priori voting power,
that {\em only} dummies have zero a priori voting power,
and that adding a dummy to a game has no effect on other players' a priori voting power.
The first dummy postulate is reasonable because a dummy is effectively a non-player;
the second is reasonable because, by definition, non-dummies are decisive in at least one possible division;
and the third is reasonable because, just as changes in the population of (literal) non-players outside of $[n]$ have no impact on players' a priori voting power,
so too should changes in the population of dummies who are members of $[n]$ have no impact.%
\footnote {It would be unreasonable to expect only dummies to have zero voting power in general,
because if only divisions in which the player is unsuccessful have positive probability,
then even a non-dummy player might have zero a posteriori voting power.
It would also be unreasonable to expect adding dummies to have no impact on others' a posteriori power,
because such power might change if their votes are correlated with the added dummy.}

More formally, let $\hat{\mathcal{G}}$ be the game formed by adding a dummy voter $d$ to $\mathcal{G}$.
A measure of voting power $\Psi$ satisfies the {\em dummy postulates} if:%
\footnote{ {\sc (dum-1)} is called the {\em dummy property} and {\sc (dum-3)} the {\em strong dummy property} in \citet[p. 87]{FelMZ98}.
See also \citet[p. 222]{FelM98}, where {\sc (dum-1)} and {\sc (dum-2)} are together called {\em vanishing just for dummies}
and {\sc (dum-3)} is called {\em ignoring dummies}.}
\begin{enumerate}
 \item [{\sc (dum-1)}] If $i$ is a dummy voter, then $\psi_i=0$.
\noindent \item[{\sc (dum-2)}] $\psi_i=0$ only if $i$ is a dummy voter.%
\noindent \item[{\sc (dum-3)}] $\psi_i=\hat{\psi}_i$ \  for all $i\neq d$.
\end{enumerate}


\begin{theorem}\label{thm:dummy-postulate}
RM satisfies the dummy postulates.
\end{theorem}

\section{The Dominance Postulate}\label{sec:dominance-postulate}
Here we consider the dominance postulate.
For any subset $S\subseteq [n]$ with $i,j\notin S$, we say that player $j$ {\em weakly dominates} player $i$ if
for any winning division $\mathbb{S}\cup \mathbbm{i}$,
the division $\mathbb{S}\cup \mathbbm{j}$ is winning.
A player {\em strictly dominates} another if the former weakly dominates the latter but not vice versa.

A measure of voting power $\Psi$ satisfies the {\em dominance postulate} if:
\begin{enumerate}
 \item[{\sc (dom-1)}] ~$\psi_j\ge \psi_i$ whenever $j$ weakly dominates $i$, and
\noindent \item[{\sc (dom-2)}] ~$\psi_j> \psi_i$ whenever $j$ strictly dominates $i$.
\end{enumerate}

Thus the dominance postulate holds that a player who is able to replace another player in any successful subset of players 
without compromising the subset's success, and who sometimes can replace that player in an unsuccessful subset 
and render it successful, ought to have greater a priori voting power than the latter (and that if two players can each replace 
the other in any successful subset without affecting the outcome, they ought to have equal voting power).
This is reasonable to expect of any measure of a priori voting power because a (strictly) dominant player is just as 
effective as the dominated player (and then some).%
\footnote{\citeauthor{FelM98}'s \citeyearpar[p. 244]{FelM98} formulation corresponds to {\sc (dom-2)}.
Our formulation is stronger,
because {\sc (dom-1)} requires that the voting power of two players who weakly dominate each other be equal.
\citeauthor{LarV05a}'s \citeyearpar{LarV05a} formulation is even weaker than Feslenthal and Machover's,
because the former merely requires that, if $j$ strictly dominates $i$, then $j$'s voting power not be less than $i$'s.
It therefore allows two players, one of whom strictly dominates another, to have equal voting power.}

 \begin{theorem}\label{thm:dominance}
RM satisfies the dominance postulate.
\end{theorem}

\section{The Donation Postulates}\label{sec:donation-postulate}

Next we study the donation postulates.
Specifically, we investigate the consequence of a player $j$ donating or transferring its vote, partially or fully, to player $i$.
But what does it mean for player $j$ to transfer its vote to $i$?
In a weighted voting game, the notion is clear: it simply means $j$ transfers (part of) the weight of its vote to $i$.
For example, in the weighted voting game $\{8:5,4,3,2\}$, the first player could fully transfer its vote's weight of $5$ to the fourth player 
whose vote is then weighted $7$.
Alternatively, the first player could transfer its voting weight to the fourth only partially;
for example, it could transfer a weight of $3$, which would leave a voting weight of $2$ for the former and $5$ for the latter.

However, since not all voting games are weighted voting games, players' votes may not have a weight that could be transferred.
Thus for simple voting games in general, we define a partial or full transfer via the {\sc yes}-successful sets $\hat{\mathcal{W}}$ in 
a modified game $\hat{\mathcal{G}}$.
Let player $j$ {\em fully donate} (or transfer) its vote to player $i$ if, for all $S$ containing neither $i$ nor $j$, 
\begin{align*}
S\cup\{i,j\} \in \hat{\mathcal{W}} &\iff S\cup\{i,j\} \in \mathcal{W} \nonumber\\
S\cup\{i\} \in \hat{\mathcal{W}} &\iff S\cup\{i,j\} \in \mathcal{W} \nonumber\\
S\cup\{j\} \in \hat{\mathcal{W}} &\iff S \in \mathcal{W} \nonumber\\
S \in \hat{\mathcal{W}} &\iff  S \in \mathcal{W} \nonumber
\end{align*}
Intuitively, this construction implies that in $\hat{\mathcal{G}}$ player $i$ has the voting influence 
of $i$ and $j$ together in $\mathcal{G}$, whereas player $j$ has the zero voting influence 
in $\hat{\mathcal{G}}$.
It is easy to verify that $j$ is a dummy voter in $\hat{\mathcal{G}}$;
therefore, $j$ has indeed fully transferred its vote to $i$.


Similarly, let player $j$ {\em partially donate} (or transfer) its vote to player $i$ if, for all $S$ containing neither $i$ nor $j$, 
\begin{align*}
S\cup\{i,j\} \in \hat{\mathcal{W}} &\iff S\cup\{i,j\} \in \mathcal{W} &&&&\nonumber\\
S\cup\{i\} \in \mathcal{W} &\Longrightarrow S\cup\{i\} \in \hat{\mathcal{W}} &\wedge&& S\cup\{i\} \in \hat{\mathcal{W}} &\Longrightarrow S\cup\{i,j\} \in \mathcal{W}\nonumber\\
S\cup\{j\} \in \hat{\mathcal{W}} &\Longrightarrow S\cup\{j\}  \in \mathcal{W}  &\wedge&& S\in \mathcal{W} &\Longrightarrow S\cup\{j\}  \in \hat{\mathcal{W}}\nonumber\\
S \in \hat{\mathcal{W}} &\iff  S \in \mathcal{W} &&&&\nonumber
\end{align*}
To parse this consider the second set of conditions. These state that $i$ cannot be less successful (at $\mathbb{S}\cup \{\mathbbm{i}\})$) after the partial transfer 
from $j$ {\bf but} cannot be more successful than if $j$ had fully transferred its vote. Similarly the third set of conditions state  
that $j$ cannot be more successful (at $\mathbb{S}\cup \{\mathbbm{j}\})$) after the partial transfer 
to $i$ {\bf but} cannot be less successful than if it had fully transferred its vote.

We can now state the donation postulates.
First, consider a modified game $\hat{\mathcal{G}}$ in which player $j$ partially transfers its vote to player $i$.
A measure of voting power $\Psi$ then satisfies the {\em partial-donation postulate}%
\footnote{Roughly equivalent to \citeauthor{FelM98}'s \citeyearpar{FelM98} {\em transfer postulate}.
``Roughly'' because they impose strict inequality, but to do so restrict the postulate to non-dummies.}
 if the a priori voting power of player $i$ in $\hat{\mathcal{G}}$ is at least 
equal to its a priori voting power in $\mathcal{G}$:
\begin{enumerate}
 \item[{\sc (don-1)}] ~$\hat{\psi}_i \ge \psi_i$
\end{enumerate}

Second, consider a modified game $\hat{\mathcal{G}}$ in which $j$ fully transfers its vote to $i$.
A measure of voting power $\Psi$ then satisfies the {\em full-donation postulate} if the a priori voting power of player $i$ in $\hat{\mathcal{G}}$ is at least 
equal to the a priori voting powers of player $i$ {\bf and} of player $j$ in $\mathcal{G}$:
\begin{enumerate}
 \item[{\sc (don-2)}] ~$\hat{\psi}_i \ge \max (\psi_i, \psi_j )$
\end{enumerate}

The full-donation postulate is much stronger than the partial-donation postulate but it
is reasonable to expect a measure of voting power to satisfy it because the 
player to whom a vote is transferred remains at least as effective as it was prior to the transfer (the player 
has not lost anything) and becomes at least as effective as the player who fully transfers its vote (the beneficiary 
gains the entirety of the donor's vote).

\begin{theorem}\label{thm:partial-donation-postulate}
$RM$ satisfies the partial-donation postulate.
\end{theorem}
\begin{theorem}\label{thm:full-donation-postulate}
$RM$ satisfies the full-donation postulate.
\end{theorem}

  \section{The Minimum-Power Bloc Postulate}
Next we consider the minimum-power bloc postulate.
We assume player $i$ and $j$ agree to form an indissoluble {\em bloc};
equivalently, player $i$ annexes the vote of player $j$. 
Again the postulate can be formulated via a modified game $\hat{\mathcal{G}}$.
Because of the annexation, $\hat{\mathcal{G}}$ has one fewer player than the original game $\mathcal{G}$.
Specifically, let $I=\{i,j\}$ denote the bloc player in $\hat{\mathcal{G}}$.
Player $i$ annexes $j$'s vote if, for all $S$ containing neither $i$ nor $j$, 
\begin{align*}
S \in \hat{\mathcal{W}} &\iff  S \in \mathcal{W} \nonumber\\
S\cup\{I\} \in \hat{\mathcal{W}} &\iff S\cup\{i, j\} \in \mathcal{W} \nonumber
\end{align*}
A measure of voting power $\Psi$ then satisfies the {\em minimum-power bloc postulate} if the a priori voting power of bloc player $I=\{i,j\}$ 
in $\hat{\mathcal{G}}$ is at least the a priori voting power of both player $i$ and of player $j$ in $\mathcal{G}$:
\begin{enumerate}
 \item[{\sc (bloc-1)}] ~$\hat{\psi}_I \ge \max (\psi_i, \psi_j )$
\end{enumerate}

Thus the minimum-power bloc postulate requires that the bloc's a priori voting power be at least as large as that of it most powerful member.
On the one hand, there is no good reason to expect a bloc's voting power to be equal to the sum of the power of its individual members
\citep[pp. 226-27]{FelM98}.
On the other hand, it is reasonable to expect satisfaction of the minimum-power bloc postulate,
because a bloc can do everything its most powerful member can \citep{AbiV22b}.
The postulate is justified for reasons similar to those justifying the donation postulate.
\begin{theorem}\label{thm:bloc-postulate}
$RM$ satisfies the minimum-power bloc postulate.
\end{theorem}

\noindent This result is not surprising.
As \citet{FelM98} explain, any measure satisfying the donation postulate (or transfer postulate)
will satisfy the minimum-power bloc and dominance postulates.

\section{The Quarrel Postulate}\label{sec:quarrel-postulate}

It would be paradoxical if ruling out effective cooperation between two players were to somehow increase their individual voting power.
A measure of voting power that displays such a ``quarrelling paradox'' can be said to violate a {\em quarrel postulate}.
It is reasonable to expect a measure of voting power to satisfy such a postulate because, manifestly,
one way that players effectively realize outcomes is by joining forces and voting together.

Yet previous attempts, by \citet[p. 237]{FelM98} and \citet[p. 30]{LarV05a},
to formulate a reasonable quarrel postulate for a priori voting power come up short.
This is because these formulations are not based on a conception of quarrelling that compares
voting games both of which respect {\em monotonicity}.
It would be unreasonable to expect a measure of voting power to satisfy a quarrel postulate for voting games that violate monotonicity,
because any paradoxical results may stem from violations of monotonicity rather than the unreasonability of the measure itself \citep{AbiV22a}.
In addition, Felsenthal and Machover's conception is asymmetric: it is a conception of quarrelling on the {\sc yes} side only.
While it may be reasonable to expect a measure of {\sc yes}-voting power to satisfy an asymmetric {\sc yes}-quarrel postulate,
it is unreasonable to expect a measure of a player's total voting power to do so,
for the simple reason that a {\sc yes}-quarrel may diminish a player's {\sc yes}-voting power but increase its {\sc no}-voting power.

We therefore base our quarrel postulate on a new conception of quarrelling that not only adequately captures
the intuitive idea of a quarrel, but is also symmetric and monotonic.%
\footnote{For a full defence of this conception and postulate, see \citet{AbiV22a}.}
In particular, we say player $i$ has a weak, symmetric quarrel with player $j$ if
{\bf both} $\mathbb{S}\cup\mathbbm{\{i,j\}}$ wins if and only if either $\mathbb{S}\cup\mathbbm{i}$ or $\mathbb{S}\cup\mathbbm{j}$ wins
{\bf and} $\mathbb{S}$ loses if and only if either $\mathbb{S}\cup\mathbbm{i}$ or $\mathbb{S}\cup\mathbbm{j}$ loses.
Thus it cannot be the case that $\mathbb{S}\cup\mathbbm{i}$ and $\mathbb{S}\cup\mathbbm{j}$ both lose and $\mathbb{S}\cup\mathbbm{\{i,j\}}$ wins;
similarly it cannot be the case that $\mathbb{S}\cup\mathbbm{i}$ and $\mathbb{S}\cup\mathbbm{j}$ win but $\mathbb{S}$ loses.
(Informally this would be $i$ and $j$ effectively cooperating and doing better than they can individually.)
To formulate the postulate,
we begin with a monotonic binary game $\mathcal{G}$,
and derive a game $\hat{\mathcal{G}}$ from $\mathcal{G}$ by inducing a quarrel between $i$ and $j$ in the specified sense.
We specify the transformation rule in terms of the {\sc yes}-successful sets $\hat{\mathcal{W}}$ and $\mathcal{W}$ for $\hat{\mathcal{G}}$ and $\mathcal{G}$, respectively.
For any $S$ containing neither $i$ nor $j$,
\begin{align*}
S\cup \{i,j\} \in \hat{\mathcal{W}} &\iff S\cup \{i\} \in \mathcal{W} \ \vee\  S\cup \{j\}  \in \mathcal{W}\\
S \in \hat{\mathcal{W}} &\iff S\cup \{i\} \in \mathcal{W} \ \wedge\  S\cup \{j\}  \in \mathcal{W}\\
S\cup \{i\} \in \hat{\mathcal{W}} & \iff  S\cup \{i\} \in \mathcal{W} \\
S\cup \{j\} \in \hat{\mathcal{W}} & \iff  S\cup \{j\} \in \mathcal{W}
\end{align*}

\noindent The first two properties state that if $i$ and $j$ vote on the same side ({\sc yes} and {\sc no}, respectively),
then the quarrel implies that a group of players with $i$ and $j$ on their side can be no more successful in $\hat{\mathcal{G}}$
than that group was with either $i$ or $j$ separately on their side in $\mathcal{G}$.
The last two properties state that the quarrel has no effect if $i$ and $j$ vote on opposite sides.

\begin{theorem}\label{thm:quarrel-monotonicity}
The modified game $\hat{\mathcal{G}}$ is monotonic.
\end{theorem}
\noindent Because this transformation rule results in a monotonic game in which $i$ and $j$ quarrel,
we should reasonably expect the a priori voting power of $i$ and $j$ not to increase.
This intuition is captured in the quarrel postulate.
A measure of voting power $\Psi$ satisfies the standard {\em quarrel postulate}
if the a priori voting power of $i$ and $j$ is not greater in $\hat{\mathcal{G}}$ than in $\mathcal{G}$:
\begin{enumerate}
 \item[{\sc (quar-1)}] ~$\hat{\psi}_i\le \psi_i$, and
\noindent \item[{\sc (quar-2)}] ~$\hat{\psi}_j\le \psi_j$.
\end{enumerate}

\begin{theorem}\label{thm:quarrel-postulate}
RM satisfies the standard quarrel postulate.
\end{theorem}

\section{Conclusion}\label{sec:conc}

This completes our analysis.
We have motivated our particular construction of a recursive measure of voting power with partial efficacy
via the lattice representation of voting games and the concept of random walks in stochastic processes,
and proven that $RM$ satisfies six sets of postulates, namely,
the iso-invariance, dummy, dominance, donation, minimum-power bloc, and quarrel postulates.
(In other, complementary work, we show that $RM$ satisfies five further, ``blocker'' postulates,
including two subadditivity blocker postulates, two minimum-power blocker postulates, and an added-blocker postulate \citep{AbiV22b}.)
Since it has elsewhere been argued that a recursive measure incorporating partial efficacy intuitively captures the core concept of voting power
better than rival measures that do not \citep{Abi22},
we conclude that $RM$ is a reasonable measure of voting power by the lights of the two-pronged approach to justification.

\section*{Appendix}

This appendix contains proofs of all results presented in the paper.

\subsection{Proofs for Section~4}

\restatetheorem{thm:iso-postulate}
   RM satisfies the iso-invariance postulate.
\end{theorem}
\begin{proof}
This follows immediately by the recursive definition of $RM$, since the {\sc yes}-posets and {\sc no}-posets for player $i$ are identical in 
$\mathcal{G}$ and $\hat{\mathcal{G}}$.
\end{proof}

To prove that $RM$ satisfies the three dummy postulates, we use the following lemma.
\begin{customlemma}{A.1}\label{lem:dummy}
Let $\hat{\alpha}$ be the players' efficacy scores in the new game $\hat{\mathcal{G}}$ formed by the addition of a dummy voter $d$ to $\mathcal{G}$.
Then, for any subset of players $S$, $d\notin {S}$, the $RM$ efficacy scores of any player~$i$ satisfy:
\begin{align*}
\hat{\alpha}_i^+(\mathbb{S}) &=\hat{\alpha}_i^+(\mathbb{S}\cup\mathbbm{d})=\alpha_i^+(\mathbb{S}) \tag{A1} \label{eq:A1}\\
\hat{\alpha}_i^-(\mathbb{S}) &=\hat{\alpha}_i^-(\mathbb{S}\cup\mathbbm{d})=\alpha_i^-(\mathbb{S}) \tag{A2} \label{eq:A2}\\
\hat{\alpha}_i(\mathbb{S}) &=\hat{\alpha}_i(\mathbb{S}\cup\mathbbm{d})=\alpha_i(\mathbb{S}) \tag{A3} \label{eq:A3}
\end{align*}
\end{customlemma}
\begin{proof}
Take any player $i$ and division $\mathbb{S}$ with {\sc yes}-efficacy score $\alpha_i^+(\mathbb{S})$ in the game $\mathcal{G}$. 
Suppose a dummy $d$ is added to create a new game $\hat{\mathcal{G}}$.

First, we want to show that  $\hat{\alpha}_i^+(\mathbb{S})=\hat{\alpha}_i^+(\mathbb{S}\cup\mathbbm{d})=\alpha_i^+(\mathbb{S})$.
The key fact is that $i$ is {\sc yes}-decisive at $\mathbb{S}$ in $\mathcal{G}$ if and only if it is {\sc yes}-decisive at 
both $\mathbb{S}$ and $\mathbb{S}\cup\mathbbm{d}$ in 
$\hat{\mathcal{G}}$.
It follows that $\hat{\alpha}_i^+(\mathbb{S})=\alpha_i^+(\mathbb{S})$, since the loyal descendants 
of $\mathbb{S}$ in the {\sc yes}-poset of the modified game form a sub-poset identical to the corresponding sub-poset, 
formed by $\mathbb{S}$'s loyal descendants, in the {\sc yes}-poset of the original game.
The properties of the corresponding random walks are then identical and so 
\begin{equation}\label{eq:Dum1}
\hat{\alpha}_i^+(\mathbb{S})=\alpha_i^+(\mathbb{S})
\end{equation}

Now we must show that $\hat{\alpha}_i^+(\mathbb{S}\cup\mathbbm{d})=\alpha_i^+(\mathbb{S})$. 
We proceed by induction on the cardinality of $S$. For the base case, $S=\emptyset$.
Now, by unanimity, $\mathbb{S}=(\emptyset, [\hat{n}])$ is losing. Since $d$ is a dummy it follows that 
the division $\mathbbm{d}$ is also losing.
Thus $\hat{\alpha}_i^+(\mathbb{S}\cup \mathbbm{d})=0=\alpha_i^+(\mathbb{S})$.

For the induction step, if $i\notin S$ or $\mathbb{S}\cup\mathbbm{d}$ is losing, 
then $\hat{\alpha}_i^+(\mathbb{S}\cup \mathbbm{d})=0=\alpha_i^+(\mathbb{S})$.
So assume $i\in S$ and $\mathbb{S}\cup\mathbbm{d}$ is a winning division.
Then
\begin{eqnarray}
\hat{\alpha}^+_i(\mathbb{S}\cup\mathbbm{d}) 
&=& \frac{1}{|\hat{LC}(\mathbb{S}\cup\mathbbm{d})|}\cdot  \sum_{\hat{\mathbb{S}}\in LC(\mathbb{S}\cup \mathbbm{d})}\hat{\alpha}^+_i(\mathbb{S}\cup\mathbbm{d\setminus k}) \nonumber \\
&=& \frac{1}{|\hat{LC}(\mathbb{S}\cup\mathbbm{d})|}\cdot  \sum_{k\in S\cup\{d\}} \hat{\alpha}^+_i(\mathbb{S}\cup\mathbbm{d\setminus k}) \nonumber \\
&=& \frac{1}{|\hat{LC}(\mathbb{S}\cup\mathbbm{d})|}\cdot \left( \sum_{k\in S} \hat{\alpha}^+_i(\mathbb{S}\cup\mathbbm{d\setminus k}) 
	+ \hat{\alpha}^+_i(\mathbb{S}) \right)\nonumber \\
&=& \frac{1}{|\hat{LC}(\mathbb{S}\cup\mathbbm{d})|}\cdot  \hat{\alpha}^+_i(\mathbb{S})  
	+ \frac{1}{|\hat{LC}(\mathbb{S}\cup\mathbbm{d})|}\cdot  \sum_{k\in S} \hat{\alpha}^+_i(\mathbb{S}\cup\mathbbm{d\setminus k})\nonumber \\
&=& \frac{1}{|\hat{LC}(\mathbb{S}\cup\mathbbm{d})|}\cdot  \alpha^+_i(\mathbb{S})  
	+ \frac{1}{|\hat{LC}(\mathbb{S}\cup\mathbbm{d})|}\cdot  \sum_{k\in S} \hat{\alpha}^+_i(\mathbb{S}\cup\mathbbm{d\setminus k})\nonumber \\
&=& \frac{1}{|\hat{LC}(\mathbb{S}\cup\mathbbm{d})|}\cdot  \alpha^+_i(\mathbb{S})  
	+ \frac{1}{|\hat{LC}(\mathbb{S}\cup\mathbbm{d})|}\cdot  \sum_{k\in S} \alpha^+_i(\mathbb{S}\setminus \mathbbm{k})\nonumber 
 \end{eqnarray}
We remark that the second equality holds because $\hat{\alpha}^+_i(\mathbb{S}\cup\mathbbm{d\setminus k})=0$ for any non-loyal child of
$\mathbb{S}\cup\mathbbm{d}$.
The fourth equality follows from (\ref{eq:Dum1}). The fifth equality follows by the induction hypothesis.
Now $|\hat{LC}(\mathbb{S}\cup\mathbbm{d})| =|\hat{LC}(\mathbb{S})|+1$ because $d$ is a dummy voter. Thus
\begin{eqnarray}
\hat{\alpha}^+_i(\mathbb{S}\cup\mathbbm{d}) 
&=& \frac{1}{|\hat{LC}(\mathbb{S}\cup\mathbbm{d})|}\cdot  \alpha^+_i(\mathbb{S})  + \frac{|\hat{LC}(\mathbb{S})|}{|\hat{LC}(\mathbb{S}\cup\mathbbm{d})|}
	\cdot \frac{1}{|\hat{LC}(\mathbb{S})|}\cdot  \sum_{k\in S} \alpha^+_i(\mathbb{S}\setminus \mathbbm{k})\nonumber \\
&=& \frac{1}{|\hat{LC}(\mathbb{S}\cup\mathbbm{d})|}\cdot  \alpha^+_i(\mathbb{S})  
	+ \frac{|\hat{LC}(\mathbb{S})|}{|\hat{LC}(\mathbb{S}\cup\mathbbm{d})|}\cdot \alpha^+_i(\mathbb{S})\nonumber \\
&=& \frac{1}{|\hat{LC}(\mathbb{S}\cup\mathbbm{d})|}\cdot  \alpha^+_i(\mathbb{S})  
	+ \left(1-\frac{1}{|\hat{LC}(\mathbb{S}\cup\mathbbm{d})|}\right) \cdot \alpha^+_i(\mathbb{S})\nonumber \\
&=& \alpha^+_i(\mathbb{S})  \nonumber
 \end{eqnarray}
 Thus (\ref{eq:A1}) holds. 
 A symmetric argument applies to show $\hat{\alpha}^-(\mathbb{S})=\hat{\alpha}^-(\mathbb{S}\cup\mathbbm{d})=\alpha^-(\mathbb{S})$ and thus (\ref{eq:A2}).
Summing (\ref{eq:A1}) and (\ref{eq:A2}), we obtain $\hat{\alpha}(\mathbb{S}) =\hat{\alpha}(\mathbb{S}\cup\mathbbm{d})=\alpha(\mathbb{S})$. So (\ref{eq:A3}) holds.
\end{proof}

\restatetheorem{thm:dummy-postulate}
RM satisfies the dummy postulates.
\end{theorem}
\begin{proof}
Take any winning division $\mathbb{S}$ where $d\in S$. Since $d$ is a dummy, it follows that $\mathbb{S}\setminus \mathbbm{d}$ is also a winning division.
Thus, $d$ is never decisive and recursively we have that $\alpha^+_d(\mathbb{S})=0$. Similarly, $\alpha^-_d(\mathbb{S})=0$.
Hence, $\alpha_d(\mathbb{S})=\alpha^+_d(\mathbb{S})+\alpha^-_d(\mathbb{S})=0$ and so a dummy voter's efficacy score is zero at any division $\mathbb{S}$.
Property ({\sc dum-1}) then holds since
\begin{equation*}
RM'_d = \frac{1}{2^n}\cdot\sum_{S\in \mathcal{D}} \alpha_d(\mathbb{S}) =0\\
 \end{equation*}

Now assume player $i$ is not a dummy. Then there exists a division $\mathbb{T}$ at which $i$ is {\sc yes}-decisive. So, by definition, $\alpha^+_i(\mathbb{T}) =1$. 
Property ({\sc dum-2}) then holds since
\begin{equation*}
RM'_i 
\ =\  \frac{1}{2^n}\cdot\sum_{S\in \mathcal{D}} \alpha_i(\mathbb{S}) 
\ \ge\ \frac{1}{2^n}\cdot\sum_{S\in \mathcal{D}} \alpha^+_i(\mathbb{S}) 
\ \ge\ \frac{1}{2^n}\cdot \alpha^+_i(\mathbb{T}) 
\ =\ \frac{1}{2^n}\
\ >\ 0
 \end{equation*}
Now consider property ({\sc dum-3}). For any voting-independent probability distribution $\mathbb{P}$, when a dummy voter $d$ is added to the game we have that
\begin{align*}
\hat{RM}_i 
&=\sum_{\mathbb{S}\in \hat{\mathcal{D}}: d\in S} \hat{\alpha}_i(\mathbb{S})\cdot \hat{\mathbb{P}}(\mathbb{S})
	+ \sum_{\mathbb{S}\in \hat{\mathcal{D}}: d\notin S} \hat{\alpha}_i(\mathbb{S})\cdot \hat{\mathbb{P}}(\mathbb{S})\\
&=\sum_{\mathbb{S}\in \hat{\mathcal{D}}: d\in S} \alpha_i(\mathbb{S}\setminus\mathbbm{d})\cdot \hat{\mathbb{P}}(\mathbb{S})
 	+ \sum_{\mathbb{S}\in \hat{\mathcal{D}}: d\notin S} \alpha_i(\mathbb{S})\cdot \hat{\mathbb{P}}(\mathbb{S})\\
&=\sum_{\mathbb{S}\in \hat{\mathcal{D}}: d\in S} \alpha_i(\mathbb{S})\cdot \mathbb{P}(d\mathrm{\ votes\ yes})\cdot \mathbb{P}(\mathbb{S}\setminus \mathbbm{d})
	+ \sum_{\mathbb{S}\in \hat{\mathcal{D}}: d\notin S} \alpha_i(\mathbb{S})\cdot  \mathbb{P}(d\mathrm{\ votes\ no})\cdot \mathbb{P}(\mathbb{S})\\
&=\sum_{\mathbb{S}\in \mathcal{D}} \alpha_i(\mathbb{S})\cdot \mathbb{P}(d\mathrm{\ votes\ yes})\cdot \mathbb{P}(\mathbb{S})
	+ \sum_{\mathbb{S}\in \mathcal{D}} \alpha_i(\mathbb{S})\cdot  \mathbb{P}(d\mathrm{\ votes\ no})\cdot \mathbb{P}(\mathbb{S})\\
&=\sum_{\mathbb{S}\in \mathcal{D}} \alpha_i(\mathbb{S})\cdot \mathbb{P}(\mathbb{S})\\
&= RM_i
\end{align*}
The second equality holds by Lemma~\ref{lem:dummy}.
The third equality follows from voting independence in $\mathbb{P}$.
Thus $RM$ satisfies ({\sc dum-3}) for any voting-independent probability distribution.
Now, our interest lies in the equiprobable division distribution
$\mathbb{P}=\frac{1}{|\mathcal{D}|}=\frac{1}{2^n}$. We remark that 
an equiprobable division distribution could violate voting independence, but
in such cases the probability distribution of divisions is identical
to an equiprobable distribution that does satisfy voting independence (and equiprobable voting).
Therefore, the proof still applies and ({\sc dum-3}) is also satisfied for a priori $RM'$.
\end{proof}

\subsection*{Proofs for Section~5}
To prove that $RM$ satisfies the dominance postulate, we use the following two lemmas.

\begin{customlemma}{A.2}\label{lem:dominance-plus}
For any subset $S\subseteq [n]$ with $i,j\notin S$, 
we have $\alpha^+_i(\mathbb{S}\cup \mathbbm{i})\le \alpha^+_j(\mathbb{S}\cup \mathbbm{j})$ 
and $\alpha^+_i(\mathbb{S}\cup \mathbbm{\{i,j\}})\le \alpha^+_j(\mathbb{S}\cup \mathbbm{\{i,j\}})$ 
whenever $j$ weakly dominates $i$.
\end{customlemma}
\begin{proof}
The proof is by induction on the cardinality of $S$. 
For the base case, $S=\emptyset$ and $|S|=0$.
By the unanimity condition, $\emptyset\notin \mathcal{W}$, so $\mathbb{S}$ is losing.
If $\mathbbm{i}=(\{i\}, [n]\setminus \{i\})$ is losing then $\alpha^+_i(\mathbbm{i})=0\le  \alpha^+_j(\mathbbm{j})$.
If $\mathbbm{i}$ is winning then, by dominance, $\mathbbm{j}$ is also winning. But then $i$ is {\sc yes}-decisive at $\mathbbm{i}$ and 
$j$ is {\sc yes}-decisive at $\mathbbm{j}$. Hence $\alpha^+_i(\mathbbm{i})=1= \alpha^+_j(\mathbbm{j})$. Thus in both cases
$\alpha^+_i(\mathbbm{i})\le \alpha^+_j(\mathbbm{j})$ as required.

Next, if $\mathbbm{\{i,j\}}$ is losing then $\alpha^+_i(\mathbbm{\{i,j\}})=0=  \alpha^+_j(\mathbbm{\{i,j\}})$.
So we may assume $\mathbbm{\{i,j\}}$ is winning. We have two cases to consider.
First, if $\mathbbm{i}$ is losing then $j$ is {\sc yes}-decisive at $\mathbbm{\{i,j\}}$ and $\alpha^+_j(\mathbbm{\{i,j\}})=1\ge  \alpha^+_i(\mathbbm{\{i,j\}})$.
Second, if $\mathbbm{i}$ is winning then, by dominance, $\mathbbm{j}$ is also winning. Thus $\mathbbm{i}$ 
and $\mathbbm{j}$ are both loyal children of $\mathbbm{\{i,j\}}$. It follows that $\alpha^+_i(\mathbbm{i})=1= \alpha^+_j(\mathbbm{\{i,j\}})$ 
and $\alpha^+_i(\mathbbm{\{i,j\}})=\frac12= \alpha^+_j(\mathbbm{\{i,j\}})$. Thus, $\alpha^+_i(\mathbbm{\{i,j\}})\le \alpha^+_j(\mathbbm{\{i,j\}})$ as required.

For the induction step, let's begin by showing 
that $\alpha^+_i(\mathbb{S}\cup \mathbbm{i})\le \alpha^+_j(\mathbb{S}\cup \mathbbm{j})$ for any subset $S\subseteq [n]$ with $i,j\notin S$.
First, assume $\mathbb{S}$ is a losing division. If $\mathbb{S}\cup \mathbbm{i}$ is losing 
then $\alpha^+_i(\mathbb{S}\cup \mathbbm{i})=0 \le \alpha^+_j(\mathbb{S}\cup \mathbbm{j})$.
If $\mathbb{S}\cup \mathbbm{i}$ is winning then $\mathbb{S}\cup \mathbbm{j}$ is also winning since player $j$ dominates player $i$.
Thus $i$ and $j$ are {\sc yes}-decisive in $\mathbb{S}\cup \mathbbm{i}$ and $\mathbb{S}\cup \mathbbm{j}$, respectively, and so, 
by definition, $\alpha^+_i(\mathbb{S}\cup \mathbbm{i})=1= \alpha^+_j(\mathbb{S}\cup \mathbbm{j})$.

Second, assume $\mathbb{S}$ is a winning division.
Then by monotonicity $\mathbb{S}\cup \mathbbm{i}$ and $\mathbb{S}\cup \mathbbm{j}$ are winning; but neither $i$ or $j$ are {\sc yes}-decisive.
Thus recursively we have
\begin{eqnarray}
\alpha^+_j(\mathbb{S}\cup \mathbbm{j}) 
&=& \frac{1}{|LC(\mathbb{S}\cup \mathbbm{j})|}\cdot \sum_{k\in S\cup \{j\}} \alpha^+_j(\mathbb{S}\cup \mathbbm{j\setminus k})  \nonumber  \\
&=& \frac{1}{|LC(\mathbb{S}\cup \mathbbm{j})|}\cdot \left( \sum_{k\in S} \alpha^+_j(\mathbb{S}\cup \mathbbm{j\setminus k}) + \alpha^+_j(\mathbb{S}) \right)\nonumber 
\end{eqnarray}

Now let $k\in S_1\subseteq S$ if $\mathbb{S}\cup \mathbbm{j\setminus k}$ and $\mathbb{S}\cup \mathbbm{i\setminus k}$ both win.
Let $k\in S_2\subseteq S$ if $\mathbb{S}\cup \mathbbm{j\setminus k}$ wins and $\mathbb{S}\cup \mathbbm{i\setminus k}$ loses.
Let $k\in S_3\subseteq S$ if $\mathbb{S}\cup \mathbbm{j\setminus k}$ and $\mathbb{S}\cup \mathbbm{i\setminus k}$ both lose.
Observe that, by dominance, there does not exist $k\in S$ such that $\mathbb{S}\cup \mathbbm{j\setminus k}$ loses 
and $\mathbb{S}\cup \mathbbm{i\setminus k}$ wins.
Thus $S=S_1\cup S_2\cup S_3$ and
\begin{eqnarray}
\alpha^+_j(\mathbb{S}\cup \mathbbm{j}) &=& \frac{1}{|LC(\mathbb{S}\cup \mathbbm{j})|}
	\cdot \left( \sum_{k\in S_1} \alpha^+_j(\mathbb{S}\cup \mathbbm{j\setminus k})
	+\sum_{k\in S_2} \alpha^+_j(\mathbb{S}\cup \mathbbm{j\setminus k}) + \alpha^+_j(\mathbb{S}) \right)\nonumber \\
 &\ge& \frac{1}{|LC(\mathbb{S}\cup \mathbbm{j})|}\cdot \left( \sum_{k\in S_1} \alpha^+_i(\mathbb{S}\cup \mathbbm{i\setminus k})
 	+\sum_{k\in S_2} \alpha^+_j(\mathbb{S}\cup \mathbbm{j\setminus k}) 
 	+ \alpha^+_j(\mathbb{S}) \right)\nonumber \\
  &=& \frac{1}{|S_1|+|S_2|+1}\cdot \left( \sum_{k\in S_1} \alpha^+_i(\mathbb{S}\cup \mathbbm{i\setminus k})
  	+\sum_{k\in S_2} \alpha^+_j(\mathbb{S}\cup \mathbbm{j\setminus k}) + \alpha^+_j(\mathbb{S}) \right)\nonumber 
\end{eqnarray}
Here the inequality follows from the induction hypothesis.
Now for $k\in S_2$ we have that $\mathbb{S}\cup \mathbbm{i\setminus k}$ loses.
By monotonicity this implies that $\mathbb{S}\setminus \mathbbm{k}$ also loses.
But since $k\in S_2$, we have that $\mathbb{S}\cup \mathbbm{j\setminus k}$ wins.
Thus $j$ is {\sc yes}-decisive at $\mathbb{S}\cup \mathbbm{j\setminus k}$.
In particular, $\alpha^+_j(\mathbb{S}\cup \mathbbm{j\setminus k})=1$.
So
\begin{eqnarray}\label{eq:dom1}
\alpha^+_j(\mathbb{S}\cup \mathbbm{j}) 
 &\ge&  \frac{1}{|S_1|+|S_2|+1}\cdot \left( \sum_{k\in S_1} \alpha^+_i(\mathbb{S}\cup \mathbbm{i\setminus k})
 	+|S_2| + \alpha^+_j(\mathbb{S}) \right)\nonumber \\
 &\ge&  \min \left[ \frac{1}{|S_2|+1}\cdot \left( \sum_{k\in S_1} \alpha^+_i(\mathbb{S}\cup \mathbbm{i\setminus k}) 
 	+ \alpha^+_j(\mathbb{S}) \right) , \frac{1}{|S_2|}\cdot |S_2| \right]  \nonumber \\
  &=&  \min \left[ \frac{1}{|S_2|+1}\cdot \left( \sum_{k\in S_1} \alpha^+_i(\mathbb{S}\cup \mathbbm{i\setminus k}) 
  	+ \alpha^+_j(\mathbb{S}) \right) , 1 \right] 
 \nonumber \\
  &=&  \frac{1}{|S_2|+1}\cdot \left( \sum_{k\in S_1} \alpha^+_i(\mathbb{S}\cup \mathbbm{i\setminus k}) + \alpha^+_j(\mathbb{S}) \right)
\end{eqnarray}

Here the second inequality follows from the mathematical fact that $\frac{A+B}{C+D}\ge \min\left(\frac{A}{C}, \frac{B}{D}\right)$ 
for positive numbers $A,B,C,D$.
On the other hand 
\begin{eqnarray}\label{eq:dom2}
\alpha^+_i(S\cup \{i\}) 
 &=& \frac{1}{|LC(\mathbb{S}\cup\mathbbm{i})|}\cdot \left( \sum_{k\in S} \alpha^+_i(\mathbb{S}\cup \mathbbm{i\setminus k}) + \alpha^+_i(\mathbb{S}) \right)\nonumber \\
 &=& \frac{1}{|LC(\mathbb{S}\cup\mathbbm{i})|}\cdot \left( \sum_{k\in S_1} \alpha^+_i(\mathbb{S}\cup \mathbbm{i\setminus k}) + \alpha^+_i(\mathbb{S}) \right)\nonumber \\
 &=& \frac{1}{|S_1|+1}\cdot \left( \sum_{k\in S_1} \alpha^+_i(\mathbb{S}\cup \mathbbm{i\setminus k})+ \alpha^+_i(\mathbb{S}) \right)
 \end{eqnarray}
Together, (\ref{eq:dom1}) and (\ref{eq:dom2}) imply that 
\begin{equation}\label{eq:SiSj}
\alpha^+_i(\mathbb{S}\cup \mathbbm{i}) \le \alpha^+_j(\mathbb{S}\cup \mathbbm{j}) 
\end{equation}
as desired.

Next let's show that $\alpha^+_i(\mathbb{S}\cup \mathbbm{\{i,j\}})\le  \alpha^+_j(\mathbb{S}\cup \mathbbm{\{i,j\}})$.
\begin{eqnarray}\label{eq:both}
\alpha^+_i(S\cup \{i,j\}) 
&=& \frac{1}{|LC(\mathbb{S}\cup \mathbbm{\{i,j\}}|}\cdot \sum_{k\in S\cup \{i,j\}} \alpha^+_i(\mathbb{S}\cup \mathbbm{\{i,j\}\setminus k})\nonumber \\
&=& \frac{1}{|LC(\mathbb{S}\cup \mathbbm{\{i,j\}}|}\cdot  \sum_{k\in S} \alpha^+_i(\mathbb{S}\cup \mathbbm{\{i,j\}\setminus k}) 
	+ \alpha^+_i(\mathbb{S}\cup \mathbbm{i}) + \alpha^+_i(\mathbb{S}\cup \mathbbm{j}) \nonumber \\
&=&\frac{1}{|LC(\mathbb{S}\cup \mathbbm{\{i,j\}}|}\cdot  \sum_{k\in S} \alpha^+_i(\mathbb{S}\cup \mathbbm{\{i,j\}\setminus k}) 
	+ \alpha^+_i(\mathbb{S}\cup \mathbbm{i}) + 0\nonumber \\
&\le& \frac{1}{|LC(\mathbb{S}\cup \mathbbm{\{i,j\}}|}\cdot \sum_{k\in S} \alpha^+_i(\mathbb{S}\cup \mathbbm{\{i,j\}\setminus k}) 
	+ \alpha^+_j(\mathbb{S}\cup \mathbbm{j})\nonumber \\
&\le& \frac{1}{|LC(\mathbb{S}\cup \mathbbm{\{i,j\}}|}\cdot  \sum_{k\in S} \alpha^+_j(\mathbb{S}\cup \mathbbm{\{i,j\}\setminus k}) 
	+ \alpha^+_j(\mathbb{S}\cup \mathbbm{j})\nonumber \\
&=&\frac{1}{|LC(\mathbb{S}\cup \mathbbm{\{i,j\}}|}\cdot \sum_{k\in S} \alpha^+_j(\mathbb{S}\cup \mathbbm{\{i,j\}\setminus k}) 
	+ \alpha^+_j(\mathbb{S}\cup \mathbbm{j})+ \alpha^+_j(\mathbb{S}\cup \mathbbm{i})\nonumber \\
&=& \frac{1}{|LC(\mathbb{S}\cup \mathbbm{\{i,j\}}|}\cdot \sum_{k\in S\cup \{i,j\}} \alpha^+_j(\mathbb{S}\cup \mathbbm{\{i,j\}\setminus k})\nonumber \\
&=& \alpha^+_j(\mathbb{S}\cup \mathbbm{\{i,j\}}) \nonumber
\end{eqnarray}
Here the first inequality follows from (\ref{eq:SiSj}). The second inequality follows from the induction hypothesis.
The equalities follow by definition of the {\sc yes}-efficacy score $\alpha^+$.
Thus $\alpha^+_i(\mathbb{S}\cup \mathbbm{\{i,j\}})\le \alpha^+_j(\mathbb{S}\cup \mathbbm{\{i,j\}})$ as desired.
\end{proof}

\begin{customlemma}{A.3}\label{lem:dominance-minus}
For any subset $S\subseteq [n]$ with $i,j\in S$, 
we have $\alpha^-_i(\mathbb{S}\setminus \mathbbm{i})\le \alpha_j(\mathbb{S}\setminus \mathbbm{j})$ 
and $\alpha^-_i(\mathbb{S}\setminus \mathbbm{\{i,j\}})\le \alpha_j(\mathbb{S}\setminus \mathbbm{\{i,j\}})$ 
whenever $j$ weakly dominates $i$.
\end{customlemma}
 \begin{proof}
 Apply a symmetric argument to that in the proof of Lemma~\ref{lem:dominance-plus}.
  \end{proof}

\restatetheorem{thm:dominance}
RM satisfies the dominance postulate.
\end{theorem}
\begin{proof}
Assume that $j$ weakly dominates $i$. We claim that 
$RM'^+_i\le RM'^+_j$.
\begin{align*}
RM^+_i(\mathcal{G}) 
&= \sum_{\mathbb{S}\in \mathcal{D}} \alpha^+_{i}(\mathbb{S})\cdot \mathbb{P(S)} \\
RM'^+_i(\mathcal{G}) 
&= \frac{1}{2^n}\cdot \sum_{\mathbb{S}\in \mathcal{D}} \alpha^+_{i}(\mathbb{S}) \\
&= \frac{1}{2^n}\cdot  \sum_{\mathbb{S}: i,j\notin S} \left(\alpha^+_{j}(\mathbb{S}\cup \mathbbm{j} ) + \alpha^+_{j}(\mathbb{S}\cup \mathbbm{\{i,j\}} ) \right) \\
&\le \frac{1}{2^n}\cdot  \sum_{\mathbb{S}: i,j\notin S} \left(\alpha^+_{i}(\mathbb{S}\cup \mathbbm{i} ) + \alpha^+_{i}(\mathbb{S}\cup \mathbbm{\{i,j\}} ) \right) \\
&= \frac{1}{2^n}\cdot \sum_{\mathbb{S}\in \mathcal{D}} \alpha^+_{j}(\mathbb{S}) \\
&= RM'^+_j(\mathcal{G}) 
\end{align*}
 Here the inequality holds by Lemma~\ref{lem:dominance-plus}.
A similar argument  applying Lemma~\ref{lem:dominance-minus} shows that $RM'^-_i\le RM'^-_j$.
Summing, we obtain that $RM'_i\le RM'_j$.
Thus {\sc (dom-1)} holds.

Next assume player $j$ strictly dominates player $i$. Then there exists a subset $S$
such that the division $\mathbb{S}\cup \mathbbm{i}$ is losing whilst $\mathbb{S}\cup \mathbbm{j}$ is winning.
Since $\mathbb{S}\cup \mathbbm{i}$ is losing we have, by definition, $\alpha^+_i(\mathbb{S}\cup \mathbbm{i})=0$.
On the other hand, since $\mathbb{S}\cup \mathbbm{j}$ is winning it follows that $\alpha^+_j(\mathbb{S}\cup \mathbbm{j})>0$.
Ergo, $\alpha^+_i(\mathbb{S}\cup \mathbbm{i})< \alpha^+_j(\mathbb{S}\cup \mathbbm{j})$.
Repeating the above argument for weak domination, the inequality now becomes strict.
Hence, $RM_i< RM_j$ and {\sc (dom-2)} holds.
\end{proof}

\subsection*{Proofs for Section~6}

Before proving that $RM$ satisfies the donation postulates (Theorems~\ref{thm:partial-donation-postulate} and \ref{thm:full-donation-postulate}),
we first show that after a partial or full donation the modified game $\hat{\mathcal{G}}$ retains monotonicity.
\begin{customlemma}{A.4}
The modified game $\hat{\mathcal{G}}$ is monotonic.
\end{customlemma}
\begin{proof}
We prove this for full donation; the case of a partial donation is similar.
Assume a violation of monotonicity is caused by division $\mathbb{S}\cup\mathbbm{i}$ winning in $\hat{\mathcal{G}}$ but 
losing in $\mathcal{G}$. By definition of $\hat{\mathcal{W}}$ this implies that $\mathbb{S}\cup\mathbbm{\{i,j\}}$ wins in $\mathcal{G}$.
Let the violation arise because $\mathbb{S}\cup\mathbbm{\{i,k\}}$ loses in $\hat{\mathcal{G}}$.
Now $k\neq j$, since $\mathbb{S}\cup\mathbbm{\{i,j\}}$ wins in $\hat{\mathcal{G}}$ given that it wins in $\mathcal{G}$.
But $\mathbb{S}\cup\mathbbm{\{i,k\}}$ has the same outcome in $\hat{\mathcal{G}}$ as $\mathbb{S}\cup\mathbbm{\{i,j,k\}}$ 
does in $\mathcal{G}$. However, given that $\mathbb{S}\cup\mathbbm{\{i,j\}}$ wins in $\mathcal{G}$, then by monotonicity so too 
does $\mathbb{S}\cup\mathbbm{\{i,j,k\}}$. Thus, by definition of a full donation, $\mathbb{S}\cup\mathbbm{\{i,k\}}$ wins in $\hat{\mathcal{G}}$, a contradiction.

On the other hand assume a violation of monotonicity is caused by $\mathbb{S}\cup\mathbbm{j}$ losing in $\hat{\mathcal{G}}$ 
but winning in $\mathcal{G}$.
Let the violation arise because $\mathbb{S}\cup\mathbbm{j\setminus k}$ wins in $\hat{\mathcal{G}}$. Thus $k\neq j$.
Since  $\mathbb{S}\cup\mathbbm{j}$ loses in $\hat{\mathcal{G}}$ this means $\mathbb{S}$ loses in both $\mathcal{G}$ and $\hat{\mathcal{G}}$.
But $\mathbb{S}\cup\mathbbm{j\setminus k}$ has the same outcome in $\hat{\mathcal{G}}$ as $\mathbb{S}\setminus \mathbbm{k}$  does in $\mathcal{G}$.
But, by monotonicity, $\mathbb{S}\setminus \mathbbm{k}$ loses in $\mathcal{G}$.
Thus, by definition of a full donation, $\mathbb{S}\cup\mathbbm{j\setminus k}$ loses in $\hat{\mathcal{G}}$, a contradiction.
\end{proof}

To prove that $RM$ satisfies the partial donation postulate, we need the following lemma.
\begin{customlemma}{A.5}\label{lem:partial-donation-plus}
Let player $j$ partially donate to player $i$. Then, for any $S\subseteq [n]$ with $i,j\notin S$, 
the efficacy scores of player $i$ satisfy:
\begin{align}
\hat{\alpha}^+_i(\mathbb{S}\cup \mathbbm{i}) 
	&\ge \alpha^+_i(\mathbb{S}\cup \mathbbm{i}) \tag{P1} \label{eq:P1}\\
\hat{\alpha}^+_i(\mathbb{S}\cup \mathbbm{\{i,j\}}) 
	&\ge \alpha^+_i(\mathbb{S}\cup \mathbbm{\{i,j\}})  \tag{P2}\label{eq:P2}
\end{align}
\end{customlemma}
\begin{proof}
We prove this by induction on the cardinality of $S$. For the base case, consider $S=\emptyset$.
If $\mathbbm{i}$ wins in $\hat{\mathcal{G}}$ then $i$ is {\sc yes}-decisive at $\mathbbm{i}$
and (\ref{eq:P1}) holds. So we may assume $\mathbbm{i}$ loses in $\hat{\mathcal{G}}$.
Then, by definition of a partial donation, $\mathbbm{i}$ also loses in $\hat{\mathcal{G}}$.
It follows that $\hat{\alpha}^+_i(\mathbb{S}\cup \mathbbm{\{i,j\}}) =0$ and  (\ref{eq:P1}) holds.
Similarly, (\ref{eq:P1}) holds if $\mathbbm{\{i,j\}}$ loses in $\hat{\mathcal{G}}$.
So assume $\mathbbm{\{i,j\}}$ wins in $\hat{\mathcal{G}}$. If $\mathbbm{\{j\}}$ loses in $\hat{\mathcal{G}}$ then
$i$ is {\sc yes}-decisive at $\{\mathbbm{i,j}\}$ in $\hat{\mathcal{G}}$
and so (\ref{eq:P2}) holds. On the other hand, if $\mathbbm{\{j\}}$ wins in $\hat{\mathcal{G}}$
then (\ref{eq:P2}) holds recursively since $\hat{\alpha}^+_i(\mathbbm{i}) 
\ge \alpha^+_i(\mathbbm{i})$ and $\hat{\alpha}^+_i(\mathbbm{j}) 
= \alpha^+_i(\mathbbm{j})=0$.

Next consider the induction step.
First let's show (\ref{eq:P1}).
Again, by definition of a partial donation, we may assume 
$\mathbb{S}\cup \mathbbm{i}$ is winning in $\hat{\mathcal{G}}$.
Furthermore, we may assume $\mathbb{S}$ is a winning division in $\hat{\mathcal{G}}$, otherwise $i$ is {\sc yes}-decisive 
at $\mathbb{S}\cup \mathbbm{i}$ and the result holds trivially. Thus,
\begin{eqnarray}
\hat{\alpha}^+_i(\mathbb{S}\cup \mathbbm{i}) 
&=& \frac{1}{|\hat{LC}(\mathbb{S}\cup \mathbbm{i})|}\cdot \left( \sum_{k\in S} \hat{\alpha}^+_i(\mathbb{S}\cup \mathbbm{i\setminus k}) 
	 +\hat{\alpha}^+_i(\mathbb{S}) \right) \nonumber \\
&=& \frac{1}{|\hat{LC}(\mathbb{S}\cup \mathbbm{i})|}\cdot \left( \sum_{k\in S} \hat{\alpha}^+_i(\mathbb{S}\cup \mathbbm{i\setminus k}) 
	 +0 \right) \nonumber 
\end{eqnarray} 
  
Let $k\in S_1\subseteq S$ if $\mathbb{S}\cup \mathbbm{i\setminus k}$ wins in both $\hat{\mathcal{G}}$ and $\mathcal{G}$.
Let $k\in S_2\subseteq S$ if $\mathbb{S}\cup \mathbbm{i\setminus k}$ wins in $\hat{\mathcal{G}}$ but loses 
in $\mathcal{G}$. 
Let $k\in S_3\subseteq S$ if $\mathbb{S}\cup \mathbbm{i\setminus k}$ loses in both $\hat{\mathcal{G}}$ and $\mathcal{G}$.
Recall, by definition of a partial donation, there does not exist $k\in S$ such that $\mathbb{S}\cup \mathbbm{i\setminus k}$ 
loses in $\hat{\mathcal{G}}$ but wins in $\mathcal{G}$. Consequently $S=S_1\cup S_2\cup S_3$.
Hence,
\begin{eqnarray}\label{eq:pdon1}
\hat{\alpha}^+_i(\mathbb{S}\cup \mathbbm{i}) 
&=& \frac{1}{|\hat{LC}(\mathbb{S}\cup\mathbbm{i})|}\cdot \left( \sum_{k\in S_1} \hat{\alpha}^+_i(\mathbb{S}\cup\mathbbm{i\setminus k}) 
	+\sum_{k\in S_2} \hat{\alpha}^+_i(\mathbb{S}\cup\mathbbm{i\setminus k}) \right) \nonumber \\
&=& \frac{1}{|\hat{LC}(\mathbb{S}\cup\mathbbm{i})|}\cdot \left( \sum_{k\in S_1} \hat{\alpha}^+_i(\mathbb{S}\cup\mathbbm{i\setminus k}) +|S_2|  \right) \nonumber \\
&\ge& \frac{1}{|\hat{LC}(\mathbb{S}\cup\mathbbm{i})|}\cdot \left( \sum_{k\in S_1} \alpha^+_j(\mathbb{S}\cup\mathbbm{i\setminus k}) +|S_2| \right) \nonumber \\
&=& \frac{1}{|S_1|+|S_2|+1}\cdot \left( \sum_{k\in S_1} \alpha^+_j(\mathbb{S}\cup\mathbbm{i\setminus k}) +|S_2| \right) \nonumber \\
 &\ge& \frac{1}{|S_1|+1}\cdot  \sum_{k\in S_1} \alpha^+_j(\mathbb{S}\cup\mathbbm{i\setminus k}) \nonumber \\
 &=& \alpha^+_i(\mathbb{S}\cup \mathbbm{i})
 \end{eqnarray} 
 To see the second equality, recall that for $k\in S_2$, 
 $\mathbb{S}\cup \mathbbm{i\setminus k}$ wins in $\hat{\mathcal{G}}$ but loses 
in $\mathcal{G}$.
 Thus, by monotonicity, $\mathbb{S}\setminus \mathbbm{k}$ loses in $\mathcal{G}$ and thus, by definition,
also loses in $\mathcal{G}$.
 Hence $i$ is {\sc yes}-decisive at $\mathbb{S}\cup\mathbbm{i\setminus k}$ 
in $\hat{\mathcal{G}}$. Consequently, $\hat{\alpha}^+_i(\mathbb{S}\cup\mathbbm{i\setminus k})=1$, for each $k\in S_2$.
The first inequality holds by the induction hypothesis. 
Thus (\ref{eq:P1}) holds.

Next consider (\ref{eq:P2}). 
If $i$ is {\sc yes}-decisive at $\mathbb{S}\cup \mathbbm{\{i,j\}}$ in $\hat{\mathcal{G}}$ then we are done.
So we may assume that both $\mathbb{S}\cup \mathbbm{\{i,j\}})$ and $\mathbb{S}\cup \mathbbm{j}$ win in $\hat{\mathcal{G}}$. 
Now
\begin{eqnarray}
\hat{\alpha}^+_i(\mathbb{S}\cup \mathbbm{\{i,j\}}) 
&=& \frac{1}{|\hat{LC}(\mathbb{S}\cup\mathbbm{\{i,j\}})|}\cdot  \sum_{k\in S\cup\{i,j\}} \hat{\alpha}^+_i(\mathbb{S}\cup \mathbbm{\{i,j\}\setminus k}) \nonumber \\
&=& \frac{1}{|\hat{LC}(\mathbb{S}\cup\mathbbm{\{i,j\}})|}\cdot \left( \sum_{k\in S} \hat{\alpha}^+_i(\mathbb{S}\cup \mathbbm{\{i,j\}\setminus k}) 
	+\hat{\alpha}^+_i(\mathbb{S}\cup \mathbbm{i})  +\hat{\alpha}^+_i(\mathbb{S}\cup \mathbbm{j}) \right) \nonumber 
 \end{eqnarray}
Recall that $\mathbb{S}\cup \mathbbm{\{i,j\}\setminus k}$ has the same outcome in $\hat{\mathcal{G}}$ and $\mathcal{G}$.
Furthermore, as $\mathbb{S}\cup \mathbbm{j}$ wins in $\hat{\mathcal{G}}$ it also wins in $\mathcal{G}$, by definition of a partial donation.
Thus we have two cases: either $|\hat{LC}(\mathbb{S}\cup\mathbbm{\{i,j\}})|=|LC(\mathbb{S}\cup\mathbbm{\{i,j\}})|$
and $\mathbb{S}\cup\mathbbm{i}$ has the same outcome in both $\hat{\mathcal{G}}$ and $\mathcal{G}$, 
or $|\hat{LC}(\mathbb{S}\cup\mathbbm{\{i,j\}})|=|LC(\mathbb{S}\cup\mathbbm{\{i,j\}})|+1$ and 
$\mathbb{S}\cup\mathbbm{i}$ wins in $\hat{\mathcal{G}}$ but loses in $\mathcal{G}$.
In the former case we have
 \begin{eqnarray}
\hat{\alpha}^+_i(\mathbb{S}\cup \mathbbm{\{i,j\}}) 
&=& \frac{1}{|LC(\mathbb{S}\cup\mathbbm{\{i,j\}})|}\cdot \left( \sum_{k\in S} \hat{\alpha}^+_i(\mathbb{S}\cup \mathbbm{\{i,j\}\setminus k}) 
	+\hat{\alpha}^+_i(\mathbb{S}\cup \mathbbm{i})  +\hat{\alpha}^+_i(\mathbb{S}\cup \mathbbm{j}) \right) \nonumber \\
&=& \frac{1}{|LC(\mathbb{S}\cup\mathbbm{\{i,j\}})|}\cdot \left( \sum_{k\in S} \hat{\alpha}^+_i(\mathbb{S}\cup \mathbbm{\{i,j\}\setminus k}) 
	+\hat{\alpha}^+_i(\mathbb{S}\cup \mathbbm{i})  +0 \right) \nonumber \\
&\ge& \frac{1}{|LC(\mathbb{S}\cup\mathbbm{\{i,j\}})|}\cdot \left( \sum_{k\in S} \alpha^+_i(\mathbb{S}\cup \mathbbm{\{i,j\}\setminus k}) 
	+\hat{\alpha}^+_i(\mathbb{S}\cup \mathbbm{i})  \right) \nonumber \\
&\ge& \frac{1}{|LC(\mathbb{S}\cup\mathbbm{\{i,j\}})|}\cdot \left( \sum_{k\in S} \alpha^+_i(\mathbb{S}\cup \mathbbm{\{i,j\}\setminus k}) 
+ \alpha^+_i(\mathbb{S}\cup \mathbbm{i}) \right) \nonumber\\
&=& \alpha^+_i(\mathbb{S}\cup \mathbbm{\{i,j\}}) \nonumber
 \end{eqnarray}
 Here the first inequality holds by the induction hypothesis. The second inequality follows from (\ref{eq:P1}). 
 In the latter case we have
 \begin{eqnarray}
\hat{\alpha}^+_i(\mathbb{S}\cup \mathbbm{\{i,j\}}) 
&=& \frac{1}{|LC(\mathbb{S}\cup\mathbbm{\{i,j\}})|+1}\cdot \left( \sum_{k\in S} \hat{\alpha}^+_i(\mathbb{S}\cup \mathbbm{\{i,j\}\setminus k}) 
	+\hat{\alpha}^+_i(\mathbb{S}\cup \mathbbm{i})  +\hat{\alpha}^+_i(\mathbb{S}\cup \mathbbm{j}) \right) \nonumber \\
&=& \frac{1}{|LC(\mathbb{S}\cup\mathbbm{\{i,j\}})|+1}\cdot \left( \sum_{k\in S} \hat{\alpha}^+_i(\mathbb{S}\cup \mathbbm{\{i,j\}\setminus k}) 
	+\hat{\alpha}^+_i(\mathbb{S}\cup \mathbbm{i})  +0 \right) \nonumber \\
&=& \frac{1}{|LC(\mathbb{S}\cup\mathbbm{\{i,j\}})|+1}\cdot \left( \sum_{k\in S} \hat{\alpha}^+_i(\mathbb{S}\cup \mathbbm{\{i,j\}\setminus k}) +1 \right) \nonumber  
\end{eqnarray}
To see the third equality, note that for this case $\mathbb{S}\cup\mathbbm{i}$ wins in $\hat{\mathcal{G}}$ but loses in $\mathcal{G}$.
But by definition, $\mathbb{S}$ has the same outcome in both $\hat{\mathcal{G}}$ and $\mathcal{G}$. Furthermore, by monotonicity,
 $\mathbb{S}$ loses in $\mathcal{G}$. Thus $i$ is {\sc yes}-decisive at $\mathbb{S}\cup \mathbbm{i}$ in $\hat{\mathcal{G}}$.
Hence $\hat{\alpha}^+_i(\mathbb{S}\cup \mathbbm{i})=1$.
So
 \begin{eqnarray}
\hat{\alpha}^+_i(\mathbb{S}\cup \mathbbm{\{i,j\}}) 	
&=& \frac{1}{|LC(\mathbb{S}\cup\mathbbm{\{i,j\}})|+1}\cdot \left( \sum_{k\in S} \hat{\alpha}^+_i(\mathbb{S}\cup \mathbbm{\{i,j\}\setminus k}) 
	+1 \right) \nonumber  \\
&=& \frac{1}{|LC(\mathbb{S}\cup\mathbbm{\{i,j\}})|+1}\cdot \left( \sum_{k\in S} \alpha^+_i(\mathbb{S}\cup \mathbbm{\{i,j\}\setminus k}) 
	+1 \right) \nonumber  \\
&\ge& \frac{1}{|LC(\mathbb{S}\cup\mathbbm{\{i,j\}})|}\cdot \sum_{k\in S} \alpha^+_i(\mathbb{S}\cup \mathbbm{\{i,j\}\setminus k})  \nonumber  \\
&=&   \alpha^+_i(\mathbb{S}\cup\mathbbm{\{i,j\}})  \nonumber   
\end{eqnarray}
 Here the second equality follows from the induction hypothesis.   
  \end{proof}
  
  \begin{customlemma}{A.6}\label{lem:partial-donation-minus}
Let player $j$ make a donation to player $i$. Then, for any $S\subseteq [n]$ with $i,j\in S$, 
the efficacy scores of player $i$ satisfies:
\begin{align*}
\hat{\alpha}^-_i(\mathbb{S}\setminus \mathbbm{\{i,j\}}) &\ge \alpha^-_i(\mathbb{S}\setminus \mathbbm{\{i,j\}})\\
\hat{\alpha}^-_i(\mathbb{S}\setminus \mathbbm{i}) &\ge \alpha^-_i(\mathbb{S}\setminus \mathbbm{i})
\end{align*}
\end{customlemma}
\begin{proof}
Apply a symmetric argument to that in the proof of Lemma~\ref{lem:partial-donation-plus}.
\end{proof}

\restatetheorem{thm:partial-donation-postulate}
$RM$ satisfies the partial-donation postulate.
\end{theorem}
\begin{proof}
We have
\begin{align*}
\hat{RM'}^+_i 
&=\frac{1}{2^n}\cdot\sum_{\mathbb{S}\in \hat{\mathcal{D}}} \hat{\alpha}^+_i(\mathbb{S})\\
&=\frac{1}{2^n}\cdot \sum_{\mathbb{S}\in \mathcal{D}} \hat{\alpha}^+_i(\mathbb{S})\\
 &=\frac{1}{2^n}\cdot \sum_{\mathbb{S}\in \mathcal{D}: i,j\notin S} \left( \hat{\alpha}^+_i(\mathbb{S}\cup \mathbbm{i})
 	+ \hat{\alpha}^+_i(\mathbb{S}\cup \mathbbm{\{i,j\}}) \right)\\
 &\ge \frac{1}{2^n}\cdot \sum_{\mathbb{S}\in \mathcal{D}: i,j\notin S} \left( \alpha^+_i(\mathbb{S}\cup \mathbbm{i})
 	+ \alpha^+_i(\mathbb{S}\cup \mathbbm{\{i,j\}}) \right)\\
 &= RM'^+_i
\end{align*}
Here the inequality holds by Lemma~\ref{lem:partial-donation-plus}.
A similar argument based upon Lemma~\ref{lem:partial-donation-minus} shows that $\hat{RM'}^-_i \ge RM'^-_i$.
It follows that $\hat{RM'}_i \ge RM'_i$.
\end{proof}

The main tool required to show that $RM$ satisfies the full-donation postulate is the following technical lemma.
\begin{customlemma}{A.7}\label{lem:donation-plus}
Let player $j$ donate to player $i$. Then, for any $S\subseteq [n]$ with $i,j\notin S$, 
the efficacy scores of player $i$ satisfy:
\begin{align}
\hat{\alpha}^+_i(\mathbb{S}\cup \mathbbm{i}) 
	&\ge \max \left( \alpha^+_i(\mathbb{S}\cup \mathbbm{i}),\alpha^+_j(\mathbb{S}\cup \mathbbm{j}) \right) \tag{C1} \label{eq:C1}\\
\hat{\alpha}^+_i(\mathbb{S}\cup \mathbbm{\{i,j\}}) 
	&\ge \max \left( \alpha^+_i(\mathbb{S}\cup \mathbbm{\{i,j\}}),\alpha^+_j(\mathbb{S}\cup \mathbbm{\{i,j\}}) \right) \tag{C2} \label{eq:C2}
\end{align}
\end{customlemma}
\begin{proof}
We prove this by induction on the cardinality of $S$. For the base case, consider $S=\emptyset$.
For (\ref{eq:C1}) observe that if either $\mathbbm{i}$ or $\mathbbm{j}$ wins in $\mathcal{G}$ then, by monotonicity, so does $\mathbbm{\{i,j\}}$.
Then, by definition, $\mathbbm{i}$ wins in $\hat{\mathcal{G}}$. But then $i$ is {\sc yes}-decisive at $\mathbbm{i}$ in $\hat{\mathcal{G}}$
and so (\ref{eq:C1}) holds.
Next, if $i$ is {\sc yes}-decisive at $\mathbbm{\{i,j\}}$ in $\hat{\mathcal{G}}$ then (\ref{eq:C2}) trivially holds.
We may hence assume that $\mathbbm{\{i,j\}}$ wins in $\hat{\mathcal{G}}$.
But  $\mathbbm{j}$ has the same outcome in $\hat{\mathcal{G}}$ as $(\emptyset, [n])$ does in $\mathcal{G}$, that is, it loses.
This means $i$ is {\sc yes}-decisive at $\mathbbm{\{i,j\}}$ in $\hat{\mathcal{G}}$ and 
so $\hat{\alpha}^+_i(\mathbb{S}\cup \mathbbm{\{i,j\}}) =1$ and  (\ref{eq:C2}) holds.

Next consider the induction step.
First we must show (\ref{eq:C1}) $\hat{\alpha}^+_i(\mathbb{S}\cup \mathbbm{i})
\ge \max \left( \alpha^+_i(\mathbb{S}\cup \mathbbm{i}),\alpha^+_j(\mathbb{S}\cup \mathbbm{j}) \right)$.
Observe that $\hat{\alpha}^+_i(\mathbb{S}\cup\mathbbm{i}) \ge \alpha^+_i(\mathbb{S}\cup\mathbbm{i})$ holds by 
Lemma~\ref{lem:partial-donation-plus}.
So it remains to show that $\hat{\alpha}^+_i(\mathbb{S}\cup\mathbbm{i}) \ge \alpha^+_j(\mathbb{S}\cup\mathbbm{i})$.
Recall that $\mathbb{S}\cup \mathbbm{i}$ wins in $\hat{\mathcal{G}}$ if and only if $\mathbb{S}\cup \mathbbm{\{i,j\}}$ wins in $\mathcal{G}$.
Thus if $\mathbb{S}\cup \mathbbm{i}$ is losing in $\hat{\mathcal{G}}$ then $\mathbb{S}\cup \mathbbm{\{i,j\}}$ is losing in $\mathcal{G}$.
Therefore, by monotonicity both $\mathbb{S}\cup \mathbbm{i}$ and $\mathbb{S}\cup \mathbbm{j}$ are losing in $\mathcal{G}$.
Thus $\hat{\alpha}^+_i(\mathbb{S}\cup \mathbbm{i})=0= \max \left( \alpha^+_i(\mathbb{S}\cup \mathbbm{i}),\alpha^+_j(\mathbb{S}\cup \mathbbm{j}) \right)$.
So we may assume that $\mathbb{S}\cup \mathbbm{i}$ is winning in $\hat{\mathcal{G}}$.
Furthermore, we may assume that $\mathbb{S}$ is a winning division in $\hat{\mathcal{G}}$, otherwise $i$ is {\sc yes}-decisive 
at $\mathbb{S}\cup \mathbbm{i}$ and the result holds trivially. Thus,
\begin{eqnarray}
\hat{\alpha}^+_i(\mathbb{S}\cup \mathbbm{i}) 
&=& \frac{1}{|\hat{LC}(\mathbb{S}\cup \mathbbm{i})|}\cdot \left( \sum_{k\in S} \hat{\alpha}^+_i(\mathbb{S}\cup \mathbbm{i\setminus k}) 
	 +\hat{\alpha}^+_i(\mathbb{S}) \right) \nonumber \\
&=& \frac{1}{|\hat{LC}(\mathbb{S}\cup \mathbbm{i})|}\cdot \left( \sum_{k\in S} \hat{\alpha}^+_i(\mathbb{S}\cup \mathbbm{i\setminus k}) 
	 +0 \right) \nonumber 
\end{eqnarray} 
  
Let $k\in S_1\subseteq S$ if $\mathbb{S}\cup \mathbbm{i\setminus k}$ and $\mathbb{S}\cup \mathbbm{j\setminus k}$ both win in the 
modified game $\hat{\mathcal{G}}$.
Let $k\in S_2\subseteq S$ if $\mathbb{S}\cup \mathbbm{i\setminus k}$ wins and $\mathbb{S}\cup \mathbbm{j\setminus k}$ loses 
in $\hat{\mathcal{G}}$. Let $k\in S_3\subseteq S$ if $\mathbb{S}\cup \mathbbm{i\setminus k}$ and $\mathbb{S}\cup \mathbbm{j\setminus k}$ 
both lose in $\hat{\mathcal{G}}$. Recall $\mathbb{S}\cup\mathbbm{i}$ has the same outcome in $\hat{\mathcal{G}}$ 
as $\mathbb{S}\cup\mathbbm{\{i,j\}}$ does in $\mathcal{G}$. Further, by monotonicity, $\mathbb{S}\cup\mathbbm{\{i,j\}}$ wins in $\mathcal{G}$ if 
either $\mathbb{S}\cup\mathbbm{i}$ or $\mathbb{S}\cup\mathbbm{j}$ does. Thus there does not exist $k\in S$ such that $\mathbb{S}\cup \mathbbm{j\setminus k}$ 
wins and $\mathbb{S}\cup \mathbbm{i\setminus k}$ loses in the modified game. Consequently $S=S_1\cup S_2\cup S_3$.
Hence,
\begin{eqnarray}\label{eq:don1}
\hat{\alpha}^+_i(\mathbb{S}\cup \mathbbm{i}) 
&=& \frac{1}{|\hat{LC}(\mathbb{S}\cup\mathbbm{i})|}\cdot \left( \sum_{k\in S_1} \hat{\alpha}^+_i(\mathbb{S}\cup\mathbbm{i\setminus k}) 
	+\sum_{k\in S_2} \hat{\alpha}^+_i(\mathbb{S}\cup\mathbbm{i\setminus k}) \right) \nonumber \\
&=& \frac{1}{|\hat{LC}(\mathbb{S}\cup\mathbbm{i})|}\cdot \left( \sum_{k\in S_1} \hat{\alpha}^+_i(\mathbb{S}\cup\mathbbm{i\setminus k}) +|S_2|  \right) \nonumber \\
&\ge& \frac{1}{|\hat{LC}(\mathbb{S}\cup\mathbbm{i})|}\cdot \left( \sum_{k\in S_1} \alpha^+_j(\mathbb{S}\cup\mathbbm{j\setminus k}) +|S_2| \right) \nonumber \\
&=& \frac{1}{|S_1|+|S_2|+1}\cdot \left( \sum_{k\in S_1} \alpha^+_j(\mathbb{S}\cup\mathbbm{j\setminus k}) +|S_2| \right) \nonumber \\
 &=& \frac{1}{|S_1|+|S_2|+1}\cdot \left( \sum_{k\in S_1} \alpha^+_j(\mathbb{S}\cup\mathbbm{j\setminus k}) +|S_2| + \alpha^+_j(\mathbb{S})  \right)
 \end{eqnarray} 
 To see the second equality, recall that for $k\in S_2$, $\mathbb{S}\cup\mathbbm{i\setminus k}$ wins and $\mathbb{S}\cup\mathbbm{j\setminus k}$ 
 loses in $\hat{\mathcal{G}}$. Thus, by monotonicity, $\mathbb{S}\setminus \mathbbm{k}$ also loses.
 Hence $i$ is {\sc yes}-decisive at $\mathbb{S}\cup\mathbbm{i\setminus k}$ 
in $\hat{\mathcal{G}}$. Hence $\hat{\alpha}^+_i(\mathbb{S}\cup\mathbbm{i\setminus k})=1$, for each $k\in S_2$.
The first inequality holds by the induction hypothesis. The last equality holds since $\alpha^+_j(\mathbb{S})=0$ given $j\notin S$.
From (\ref{eq:don1}) we then obtain
 \begin{eqnarray}\label{eq:don2}
\hat{\alpha}^+_i(\mathbb{S}\cup \mathbbm{i})    
 &\ge& \min \left(\frac{|S_2|}{|S_2|}, \frac{1}{|S_1|+1}\cdot \left( \sum_{k\in S_1} \alpha^+_j(\mathbb{S}\cup\mathbbm{j\setminus k}) 
 	+ \alpha^+_j(\mathbb{S}) \right)\right) \nonumber \\
 &=& \min \left(1, \frac{1}{\hat{LC}(\mathbb{S}\cup\mathbbm{j})}   \cdot \left( \sum_{k\in S_1} \alpha^+_j(\mathbb{S}\cup\mathbbm{j\setminus k}) 
 	+ \alpha^+_j(\mathbb{S}) \right) \right) \nonumber \\
 &\ge& \min \left(1, \frac{1}{LC(\mathbb{S}\cup\mathbbm{j})}   \cdot \left( \sum_{k\in S_1} \alpha^+_j(\mathbb{S}\cup\mathbbm{j\setminus k}) 
 	+ \alpha^+_j(\mathbb{S}) \right) \right) \nonumber \\
 &=& \frac{1}{LC(\mathbb{S}\cup\mathbbm{j})}\cdot \left( \sum_{k\in S_1} \alpha^+_j(\mathbb{S}\cup\mathbbm{j\setminus k}) 
 	+ \alpha^+_j(\mathbb{S}) \right) \nonumber \\
&=&   \alpha^+_j(\mathbb{S}\cup\mathbbm{j}) 
\end{eqnarray} 
Here the first inequality again follows from the
 mathematical fact that $\frac{A+B}{C+D}\ge \min\left(\frac{A}{C}, \frac{B}{D}\right)$ for positive numbers $A,B,C,D$.
The second inequality holds because $LC(\mathbb{S}\cup\mathbbm{j})\ge \hat{LC}(\mathbb{S}\cup\mathbbm{j})$.
This follows from the fact $\mathbb{S}\cup\mathbbm{j}$ has the same outcome in $\hat{\mathcal{G}}$ as $\mathbb{S}$ does in $\mathcal{G}$; 
so if $\mathbb{S}\cup\mathbbm{j \setminus k }$ is winning in $\hat{\mathcal{G}}$ then it must also win in $\mathcal{G}$.

Thus, we have $\hat{\alpha}^+_i(\mathbb{S}\cup \mathbbm{i} )\ge 
\max \left( \alpha^+_i(\mathbb{S}\cup \mathbbm{i}),\alpha^+_j(S\cup\{j\}) \right)$.
Ergo (\ref{eq:C1}) is satisfied.

Next consider (\ref{eq:C2}). 
We must prove $\hat{\alpha}^+_i(\mathbb{S}\cup \mathbbm{\{i,j\}}) \ge \max \left( \alpha^+_i(\mathbb{S}\cup \mathbbm{\{i,j\}}),\alpha^+_j(S\cup\{i,j\}) \right)$.
Observe that $\hat{\alpha}^+_i(\mathbb{S}\cup \mathbbm{\{i,j\}}) \ge \alpha^+_i(\mathbb{S}\cup \mathbbm{\{i,j\}})$ holds by 
Lemma~\ref{lem:partial-donation-plus}. So it remains to show that $\hat{\alpha}^+_i(\mathbb{S}\cup \mathbbm{\{i,j\}}) \ge \alpha^+_j(\mathbb{S}\cup \mathbbm{\{i,j\}})$.
If $i$ is {\sc yes}-decisive at $\mathbb{S}\cup \mathbbm{\{i,j\}}$ in $\hat{\mathcal{G}}$ then we are done.
Thus we may assume that both $\mathbb{S}\cup \mathbbm{\{i,j\}})$ and $\mathbb{S}\cup \mathbbm{j}$ win in $\hat{\mathcal{G}}$. But, by definition,
$\mathbb{S}\cup \mathbbm{j}$ and $\mathbb{S}$ have the same outcome in $\hat{\mathcal{G}}$. Thus $\mathbb{S}$ also wins in $\hat{\mathcal{G}}$.
Now
\begin{eqnarray}
\hat{\alpha}^+_i(\mathbb{S}\cup \mathbbm{\{i,j\}}) 
&=& \frac{1}{|\hat{LC}(\mathbb{S}\cup\mathbbm{\{i,j\}})|}\cdot  \sum_{k\in S\cup\{i,j\}} \hat{\alpha}^+_i(\mathbb{S}\cup \mathbbm{\{i,j\}\setminus k}) \nonumber \\
&=& \frac{1}{|\hat{LC}(\mathbb{S}\cup\mathbbm{\{i,j\}})|}\cdot \left( \sum_{k\in S} \hat{\alpha}^+_i(\mathbb{S}\cup \mathbbm{\{i,j\}\setminus k}) 
	+\hat{\alpha}^+_i(\mathbb{S}\cup \mathbbm{i})  +\hat{\alpha}^+_i(\mathbb{S}\cup \mathbbm{j}) \right) \nonumber 
 \end{eqnarray}
Recall that $\mathbb{S}\cup \mathbbm{j}$ has the same outcome in $\hat{\mathcal{G}}$ as $\mathbb{S}$ does in $\mathcal{G}$.
Thus we have two cases: either $|\hat{LC}(\mathbb{S}\cup\mathbbm{\{i,j\}})|=|LC(\mathbb{S}\cup\mathbbm{\{i,j\}})|$
and $\mathbb{S}\cup\mathbbm{i}$ has the same outcome in both $\hat{\mathcal{G}}$ and $\mathcal{G}$, 
or $|\hat{LC}(\mathbb{S}\cup\mathbbm{\{i,j\}})|=|LC(\mathbb{S}\cup\mathbbm{\{i,j\}})|+1$ and 
$\mathbb{S}\cup\mathbbm{i}$ wins in $\hat{\mathcal{G}}$ but loses in $\mathcal{G}$.
In the former case we have
 \begin{eqnarray}
\hat{\alpha}^+_i(\mathbb{S}\cup \mathbbm{\{i,j\}}) 
&=& \frac{1}{|LC(\mathbb{S}\cup\mathbbm{\{i,j\}})|}\cdot \left( \sum_{k\in S} \hat{\alpha}^+_i(\mathbb{S}\cup \mathbbm{\{i,j\}\setminus k}) 
	+\hat{\alpha}^+_i(\mathbb{S}\cup \mathbbm{i})  +\hat{\alpha}^+_i(\mathbb{S}\cup \mathbbm{j}) \right) \nonumber \\
&=& \frac{1}{|LC(\mathbb{S}\cup\mathbbm{\{i,j\}})|}\cdot \left( \sum_{k\in S} \hat{\alpha}^+_i(\mathbb{S}\cup \mathbbm{\{i,j\}\setminus k}) 
	+\hat{\alpha}^+_i(\mathbb{S}\cup \mathbbm{i})  +0 \right) \nonumber \\
&\ge& \frac{1}{|LC(\mathbb{S}\cup\mathbbm{\{i,j\}})|}\cdot \left( \sum_{k\in S} \alpha^+_j(\mathbb{S}\cup \mathbbm{\{i,j\}\setminus k}) 
	+\hat{\alpha}^+_i(\mathbb{S}\cup \mathbbm{i})  +0 \right) \nonumber \\
&\ge& \frac{1}{|LC(\mathbb{S}\cup\mathbbm{\{i,j\}})|}\cdot \left( \sum_{k\in S} \alpha^+_j(\mathbb{S}\cup \mathbbm{\{i,j\}\setminus k}) 
	+\max \left( \alpha^+_i(\mathbb{S}\cup \mathbbm{i}), \alpha^+_j(\mathbb{S}\cup \mathbbm{j})\right) \right) \nonumber\\
&\ge& \frac{1}{|LC(\mathbb{S}\cup\mathbbm{\{i,j\}})|}\cdot \left( \sum_{k\in S} \alpha^+_j(\mathbb{S}\cup \mathbbm{\{i,j\}\setminus k}) 
	+\alpha^+_j(\mathbb{S}\cup \mathbbm{j})\right) \nonumber \\
&=& \alpha^+_j(\mathbb{S}\cup \mathbbm{\{i,j\}}) \nonumber
 \end{eqnarray}
 Here the first inequality holds by the induction hypothesis. The second inequality follows from (\ref{eq:C1}). 
 In the latter case we have
 \begin{eqnarray}
\hat{\alpha}^+_i(\mathbb{S}\cup \mathbbm{\{i,j\}}) 
&=& \frac{1}{|LC(\mathbb{S}\cup\mathbbm{\{i,j\}})|+1}\cdot \left( \sum_{k\in S} \hat{\alpha}^+_i(\mathbb{S}\cup \mathbbm{\{i,j\}\setminus k}) 
	+\hat{\alpha}^+_i(\mathbb{S}\cup \mathbbm{i})  +\hat{\alpha}^+_i(\mathbb{S}\cup \mathbbm{j}) \right) \nonumber \\
&=& \frac{1}{|LC(\mathbb{S}\cup\mathbbm{\{i,j\}})|+1}\cdot \left( \sum_{k\in S} \hat{\alpha}^+_i(\mathbb{S}\cup \mathbbm{\{i,j\}\setminus k}) 
	+\hat{\alpha}^+_i(\mathbb{S}\cup \mathbbm{i})  +0 \right) \nonumber \\
&=& \frac{1}{|LC(\mathbb{S}\cup\mathbbm{\{i,j\}})|+1}\cdot \left( \sum_{k\in S} \hat{\alpha}^+_i(\mathbb{S}\cup \mathbbm{\{i,j\}\setminus k}) +1 \right) \nonumber  
\end{eqnarray}
To see the third equality, note that for this case $\mathbb{S}\cup\mathbbm{i}$ wins in $\hat{\mathcal{G}}$ but loses in $\mathcal{G}$.
But by assumption, $\mathbb{S}$ has the same outcome in both $\hat{\mathcal{G}}$ and $\mathcal{G}$. Furthermore, by monotonicity,
 $\mathbb{S}$ loses in $\mathcal{G}$. Thus $i$ is {\sc yes}-decisive at $\mathbb{S}\cup \mathbbm{i}$ in $\hat{\mathcal{G}}$.
Hence $\hat{\alpha}^+_i(\mathbb{S}\cup \mathbbm{i})=1$.
So
 \begin{eqnarray}
\hat{\alpha}^+_i(\mathbb{S}\cup \mathbbm{\{i,j\}}) 	
&=& \frac{1}{|LC(\mathbb{S}\cup\mathbbm{\{i,j\}})|+1}\cdot \left( \sum_{k\in S} \hat{\alpha}^+_i(\mathbb{S}\cup \mathbbm{\{i,j\}\setminus k}) 
	+1 \right) \nonumber  \\
&\ge& \frac{1}{|LC(\mathbb{S}\cup\mathbbm{\{i,j\}})|+1}\cdot \left( \sum_{k\in S} \alpha^+_j(\mathbb{S}\cup \mathbbm{\{i,j\}\setminus k}) 
	+1 \right) \nonumber  \\
&\ge& \min\left( \frac{\sum_{k\in S} \alpha^+_j(\mathbb{S}\cup \mathbbm{\{i,j\}\setminus k}) }{|LC(\mathbb{S}\cup\mathbbm{\{i,j\}})|},
	 \frac{1}{1} \right) \nonumber  \\	
&=& \min\left(  \alpha^+_j(\mathbb{S}\cup\mathbbm{\{i,j\}}) ,
	 1 \right) \nonumber  \\	
&=&   \alpha^+_j(\mathbb{S}\cup\mathbbm{\{i,j\}})  \nonumber   
\end{eqnarray}
 Here the first inequality follows from the induction hypothesis. This completes the proof.
  \end{proof}
  
  \begin{customlemma}{A.8}\label{lem:donation-minus}
Let player $j$ make a donation to player $i$. Then, for any $S\subseteq [n]$ with $i,j\in S$, 
the efficacy scores of player $i$ satisfies:
\begin{align*}
\hat{\alpha}^-_i(\mathbb{S}\setminus \mathbbm{\{i,j\}}) &\ge \max \left( \alpha^-_i(\mathbb{S}\setminus \mathbbm{\{i,j\}}),\alpha^-_j(\mathbb{S}\setminus \mathbbm{\{i,j\}}) \right)\\
\hat{\alpha}^-_i(\mathbb{S}\setminus \mathbbm{i}) &\ge \max \left( \alpha^-_i(\mathbb{S}\setminus \mathbbm{i}),\alpha^-_j(\mathbb{S}\setminus \mathbbm{j}) \right)
\end{align*}
\end{customlemma}
\begin{proof}
Apply a symmetric argument to that in the proof of Lemma~\ref{lem:donation-plus}.
\end{proof}

\restatetheorem{thm:full-donation-postulate}
$RM$ satisfies the full-donation postulate.
\end{theorem}
\begin{proof}
We have
\begin{align*}
\hat{RM'}^+_i 
&=\frac{1}{2^n}\cdot\sum_{\mathbb{S}\in \hat{\mathcal{D}}} \hat{\alpha}^+_i(\mathbb{S})\\
&=\frac{1}{2^n}\cdot \sum_{\mathbb{S}\in \mathcal{D}} \hat{\alpha}^+_i(\mathbb{S})\\
 &=\frac{1}{2^n}\cdot \sum_{\mathbb{S}\in \mathcal{D}: i,j\notin S} \left( \hat{\alpha}^+_i(\mathbb{S}\cup \mathbbm{i})
 	+ \hat{\alpha}^+_i(\mathbb{S}\cup \mathbbm{\{i,j\}}) \right)\\
 &\ge \frac{1}{2^n}\cdot \sum_{\mathbb{S}\in \mathcal{D}: i,j\notin S} \left( \max\{ \alpha^+_i(\mathbb{S}\cup \mathbbm{i}),  \alpha^+_j(\mathbb{S}\cup \mathbbm{j})\} 
 	+ \max\{ \alpha^+_i(\mathbb{S}\cup \mathbbm{\{i,j\}}), \alpha^+_j(\mathbb{S}\cup \mathbbm{\{i,j\}}) \} \right)\\
  &\ge \frac{1}{2^n}\cdot \sum_{\mathbb{S}\in \mathcal{D}: i,j\notin S} \left( \max\{ \alpha^+_i(\mathbb{S}\cup \mathbbm{i})
  	+  \alpha^+_i(\mathbb{S}\cup \mathbbm{\{i,j\}}) ,\alpha^+_j(\mathbb{S}\cup \mathbbm{j})+ \alpha^+_j(\mathbb{S}\cup \mathbbm{\{i,j\}}) \} \right)\\
 &= \frac{1}{2^n}\cdot \max \{ \sum_{\mathbb{S}\in \mathcal{D}: i,j\notin S} \alpha^+_i(\mathbb{S}\cup \mathbbm{i})
 	+  \alpha^+_i(\mathbb{S}\cup \mathbbm{\{i,j\}}),  \sum_{\mathbb{S}\in \mathcal{D}: i,j\notin S}  \alpha^+_j(\mathbb{S}\cup \mathbbm{j})+ \alpha_j(\mathbb{S}\cup \mathbbm{\{i,j\}}) \} \\
&= \max \{RM'^+_i, RM'^+_j \}
\end{align*}
Here the first inequality holds by Lemma~\ref{lem:donation-plus}.
The second inequality holds because $\max (A_1,B_1) + \max (A_2,B_2 ) \ge A_1+A_2$ and  $\max (A_1,B_1) + \max (A_2,B_2 ) \ge B_1+B_2$.
A similar argument based upon Lemma~\ref{lem:donation-minus} shows that $\hat{RM'}^-_i \ge \max \{RM'^-_i, RM'^-_j\}$. It follows that 
$\hat{RM'}_i \ge \max \{RM'_i, RM'_j\}$.
\end{proof}

\subsection*{Proofs for Section~7}

\restatetheorem{thm:bloc-postulate}
$RM$ satisfies the minimum-power bloc postulate.
\end{theorem}
\begin{proof}
The annexation of player $j$ by player $i$ can be viewed as a two-step process.
First, $j$'s vote is donated to $i$ to give a game $\bar{\mathcal{G}}$. 
By Theorem~\ref{thm:full-donation-postulate}, $\hat{RM'}_i\ge \max (RM'_i, RM'_j)$.
Second, $j$ is removed from the game to give a game $\hat{\mathcal{G}}$.
The key observation here is that after $j$ donates to $i$, now called $I$, then $j$ becomes a dummy voter in 
$\bar{\mathcal{G}}$. Thus, by Theorem~\ref{thm:dummy-postulate}, $\hat{RM'}_I\ge \hat{RM'_i}$.
Combining these two inequalities we obtain $\hat{RM'}_I\ge \max (RM'_i, RM'_j)$.
So the minimum-power bloc postulate is satisfied.
\end{proof}

\subsection*{Proofs for Section~8}

To prove that $RM$ satisfies the quarrel postulate, we must recall that the conception of quarrelling on which we base our postulate
preserves monotonicity in the derived game $\hat{\mathcal{G}}$ in which $i$ and $j$ quarrel.
\restatetheorem{thm:quarrel-monotonicity}
The modified game $\hat{\mathcal{G}}$ is monotonic.
\end{theorem}
\noindent The proof of this lemma is provided in \citet{AbiV22a}.
Next we need the following two lemmas.

\begin{customlemma}{A.9}\label{lem:quarrel-plus}
For any division $\mathbb{S}$ with $i,j\notin S$, the efficacy scores of player $i$ satisfy:
\begin{align}
\hat{\alpha}^+_i(\mathbb{S}) &= \alpha^+_i(\mathbb{S})  \tag{D1}\label{eq:D1} \\ 
\hat{\alpha}^+_i(\mathbb{S}\cup\mathbbm{j}) &= \alpha^+_i(\mathbb{S}\cup\mathbbm{j}) \tag{D2}\label{eq:D2} \\
\hat{\alpha}^+_i(\mathbb{S}\cup\mathbbm{i}) &\le \alpha^+_i(\mathbb{S}\cup\mathbbm{i})\tag{D3}\label{eq:D3}  \\
\hat{\alpha}^+_i(\mathbb{S}\cup\mathbbm{\{i,j\}} ) &\le \alpha^+_i(\mathbb{S}\cup\mathbbm{\{i,j\}}) \tag{D4}\label{eq:D4} 
\end{align}
\end{customlemma}
\begin{proof}
Take any $S$ not containing $i$ or $j$. 
Since $i\notin S$, we have that $\hat{\alpha}^+_i(\mathbb{S}) = 0 = \alpha^+_i(\mathbb{S})$ 
and $\hat{\alpha}^+_i(\mathbb{S}\cup\mathbbm{j}) = 0= \alpha^+_i(\mathbb{S}\cup\mathbbm{j})$.
So (\ref{eq:D1}) and (\ref{eq:D2}) hold.

We now proceed by induction on $|S|$. The base case is $S=\emptyset$.
Now if $\mathbbm{i}$ wins then ${\alpha}^+_i(\mathbbm{\emptyset\cup \{i\}})=1\ge \hat{\alpha}^+_i(\mathbbm{\emptyset \cup \{i,j\}})$ as required by (\ref{eq:D3}).
If $\mathbbm{\{i,j\}}$ loses in $\hat{\mathcal{G}}$ then ${\alpha}^+_i(\mathbbm{\emptyset\cup \{i\}})\ge \hat{\alpha}^+_i(\mathbbm{\emptyset \cup \{i,j\}}) =0$.
So we may assume $\mathbbm{\{i,j\}}$ wins in $\hat{\mathcal{G}}$ and thus also in $\mathcal{G}$.
Furthermore, we may assume $\mathbbm{j}$ wins in $\mathcal{G}$, otherwise $i$ is {\sc yes}-decisive at $\mathbbm{\{i,j\}}$ in $\mathcal{G}$
and (\ref{eq:D4}) holds trivially.
We then have two possibilities: $\mathbbm{i}$
either wins in both $\mathcal{G}$ and $\hat{\mathcal{G}}$
or loses in both $\mathcal{G}$ and $\hat{\mathcal{G}}$.
In either case the sub-lattices below $\mathbbm{\{i,j\}}$ are identical for $\mathcal{D}$ and $\hat{\mathcal{D}}$ 
and thus $\hat{\alpha}^+_i(\mathbb{S}\cup\mathbbm{\{i,j\}} ) =\alpha^+_i(\mathbb{S}\cup\mathbbm{\{i,j\}})$. Hence (\ref{eq:D4}) also 
holds for the base case.

For the induction step, let's begin with showing (\ref{eq:D3}).
Now $\mathbb{S}\cup\mathbbm{i}$ wins in $\hat{\mathcal{G}}$; otherwise $\hat{\alpha}^+_i(\mathbb{S}\cup\mathbbm{i})=0$ and we are trivially done.
Thus $\mathbb{S}\cup\mathbbm{i}$ also wins in $\mathcal{G}$; this implies, by definition of $\hat{\mathcal{G}}$, that $\mathbb{S}$ has the same outcome in both
$\mathcal{G}$ and $\hat{\mathcal{G}}$. If $\mathbb{S}$ loses then $i$ is {\sc yes}-decisive at $\mathbbm{\{i,j\}}$ in both games and 
$\hat{\alpha}^+_i(\mathbb{S}\cup\mathbbm{i}) =1= \alpha^+_i(\mathbb{S}\cup\mathbbm{i})$.
Thus we may assume $\mathbb{S}$ wins and  $\hat{\alpha}^+_i(\mathbb{S}\cup\mathbbm{i})$.
Now, $\alpha^+_i(\mathbb{S}\cup\mathbbm{i})$ is defined recursively. 
For any child $\mathbb{S}\cup\mathbbm{i \setminus k }$ of $\mathbb{S}\cup\mathbbm{i}$, $i$ but not $j$ is a {\sc yes}-voter;
thus by definition of $\hat{\mathcal{G}}$, the division outcome is the same in both $\mathcal{G}$ and $\hat{\mathcal{G}}$.
Thus the set of loyal children of $\mathbb{S}\cup\mathbbm{i}$ are also identical in both $\mathcal{G}$ and $\hat{\mathcal{G}}$.
It follows that
 \begin{eqnarray}\label{eq:quar1}
\alpha^+_i(\mathbb{S}\cup\mathbbm{i}) 
&=& \frac{1}{|LC(\mathbb{S}\cup\mathbbm{i})|}\cdot \left( \sum_{k\in S} \alpha^+_i(\mathbb{S}\cup\mathbbm{i\setminus k}) 
	+ \alpha^+_i(\mathbb{S})  \right) \nonumber \\
&=& \frac{1}{|LC(\mathbb{S}\cup\mathbbm{i})|}\cdot \left( \sum_{k\in S} \alpha^+_i(\mathbb{S}\cup\mathbbm{i\setminus k}) 
	+ 0)  \right) \nonumber \\
&=& \frac{1}{|\hat{LC}(\mathbb{S}\cup\mathbbm{i})|}\cdot \left( \sum_{k\in S} \alpha^+_i(\mathbb{S}\cup\mathbbm{i\setminus k}) +0 \right) \nonumber \\
&=& \frac{1}{|\hat{LC}(\mathbb{S}\cup\mathbbm{i})|}\cdot \left( \sum_{k\in S} \alpha^+_i(\mathbb{S}\cup\mathbbm{i\setminus k}) 
	+ \hat{\alpha}^+_i(\mathbb{S}) \right) \nonumber \\
&\ge& \frac{1}{|\hat{LC}(\mathbb{S}\cup\mathbbm{i})|}\cdot \left( \sum_{k\in S} \hat{\alpha}^+_i(\mathbb{S}\cup\mathbbm{i\setminus k}) 
	+ \hat{\alpha}^+_i(\mathbb{S}) \right) \nonumber \\
&=& \hat{\alpha}^+_i(\mathbb{S}\cup\mathbbm{i}) 
 \end{eqnarray}
Here the inequality holds by the induction hypothesis. Thus $\alpha^+_i(\mathbb{S}\cup\mathbbm{i}) \ge \hat{\alpha}^+_i(\mathbb{S}\cup\mathbbm{i})$ 
and (\ref{eq:D3}) holds.

Next consider $\mathbb{S}\cup\mathbbm{\{i,j\}}$. Again, we may assume $\mathbb{S}\cup\mathbbm{\{i,j\}}$ wins in $\hat{\mathcal{G}}$ and 
thus also in $\mathcal{G}$. In addition, we may assume that $\mathbb{S}\cup\mathbbm{j}$ wins in $\mathcal{G}$ and thus also in $\hat{\mathcal{G}}$;
otherwise $i$ is {\sc yes}-decisive at $\mathbb{S}\cup\mathbbm{\{i,j\}}$ in $\mathcal{G}$ and we are done.


Let $k\in S_1\subseteq S$ if $\mathbb{S}\cup\mathbbm{\{i,j\}\setminus k}$ wins in both $\mathcal{G}$ and $\hat{\mathcal{G}}$.
Let $k\in S_2\subseteq S$ if $\mathbb{S}\cup\mathbbm{\{i,j\}\setminus k}$ wins in $\mathcal{G}$ but loses in $\hat{\mathcal{G}}$.
Let $k\in S_3\subseteq S$ if $\mathbb{S}\cup\mathbbm{\{i,j\}\setminus k}$  loses in both $\mathcal{G}$ and $\hat{\mathcal{G}}$.
Recall that, by definition of $\hat{\mathcal{G}}$, if $\mathbb{T}$ loses in $\mathcal{G}$ but wins in $\hat{\mathcal{G}}$ then $T$ contains neither $i$ nor $j$.
Thus, there does not exist $k\in S$ such that $\mathbb{S}\cup\mathbbm{\{i,j\}\setminus k}$  loses in $\mathcal{G}$  but wins in $\hat{\mathcal{G}}$. So
 \begin{eqnarray}
\alpha^+_i(\mathbb{S}\cup\mathbbm{\{i,j\}}) 
&=& \frac{1}{|LC(\mathbb{S}\cup\mathbbm{\{i,j\}})|}\cdot \left( \sum_{k\in S} \alpha^+_i(\mathbb{S}\cup\mathbbm{\{i,j\}\setminus k}) 
	+ \alpha^+_i(\mathbb{S}\cup\mathbbm{i})  + \alpha^+_i(\mathbb{S}\cup\mathbbm{j}) \right) \nonumber \\
&=& \frac{1}{|LC(\mathbb{S}\cup\mathbbm{\{i,j\}})|}\cdot \left( \sum_{k\in S} \alpha^+_i(\mathbb{S}\cup\mathbbm{\{i,j\}\setminus k}) 
	+ \alpha^+_i(\mathbb{S}\cup\mathbbm{i})  +0 \right) \nonumber \\
&=& \frac{1}{|LC(\mathbb{S}\cup\mathbbm{\{i,j\}})|}\cdot \left( \sum_{k\in S_1} \alpha^+_i(\mathbb{S}\cup\mathbbm{\{i,j\}\setminus k}) +\sum_{k\in S_2} \alpha^+_i(\mathbb{S}\cup\mathbbm{\{i,j\}\setminus k}) 
	+\alpha^+_i(\mathbb{S}\cup\mathbbm{i}) \right) \nonumber 
 \end{eqnarray}
For $k\in S_2$ we have that $\mathbb{S}\cup\mathbbm{\{i,j\}\setminus k}$ wins in $\mathcal{G}$  but loses in $\hat{\mathcal{G}}$.
But, by definition of $\hat{\mathcal{G}}$, this implies that both $\mathbb{S}\cup\mathbbm{i\setminus k}$ and $\mathbb{S}\cup\mathbbm{j\setminus k}$ 
lose in the original game $\mathcal{G}$.
In particular, $i$ is decisive at $\mathbb{S}\cup\mathbbm{\{i,j\}\setminus k}$ in the original game; consequently, 
$\alpha^+_i(\mathbb{S}\cup\mathbbm{\{i,j\}\setminus k})=1$.
Hence
\begin{eqnarray}
\alpha^+_i(\mathbb{S}\cup\mathbbm{\{i,j\}}) 
&=& \frac{1}{|LC(\mathbb{S}\cup\mathbbm{\{i,j\}})|}\cdot \left( \sum_{k\in S_1} \alpha^+_i(\mathbb{S}\cup\mathbbm{\{i,j\}\setminus k}) +|S_2|
	+\alpha^+_i(\mathbb{S}\cup\mathbbm{i}) \right) \nonumber \\
&=& \frac{1}{|S_1|+|S_2|+2}\cdot \left( \sum_{k\in S_1} \alpha^+_i(\mathbb{S}\cup\mathbbm{\{i,j\}\setminus k}) +|S_2|
	+\alpha^+_i(\mathbb{S}\cup\mathbbm{i}) \right) \nonumber \\
&\ge& \frac{1}{|S_1|+|S_2|+2}\cdot \left( \sum_{k\in S_1} \alpha^+_i(\mathbb{S}\cup\mathbbm{\{i,j\}\setminus k}) +|S_2|
	+\hat{\alpha}^+_i(\mathbb{S}\cup\mathbbm{i}) \right) \nonumber \\
&\ge& \frac{1}{|S_1|+|S_2|+2}\cdot \left( \sum_{k\in S_1} \hat{\alpha}^+_i(\mathbb{S}\cup\mathbbm{\{i,j\}\setminus k}) +|S_2|
	+\hat{\alpha}^+_i(\mathbb{S}\cup\mathbbm{i}) \right) \nonumber \\
&=& \frac{1}{|S_1|+|S_2|+2}\cdot \left( \sum_{k\in S_1} \hat{\alpha}^+_i(\mathbb{S}\cup\mathbbm{\{i,j\}\setminus k}) +|S_2|
	+\hat{\alpha}^+_i(\mathbb{S}\cup\mathbbm{i}) +\hat{\alpha}^+_i(\mathbb{S}\cup\mathbbm{j}) \right) \nonumber \\
&\ge& \min \left(\frac{|S_2|}{|S_2|}, \frac{1}{|S_1|+2}\cdot \left( \sum_{k\in S_1} \hat{\alpha}^+_i(\mathbb{S}\cup\mathbbm{\{i,j\}\setminus k}) 
	+ \hat{\alpha}^+_i(\mathbb{S}\cup\mathbbm{i})+\hat{\alpha}^+_i(\mathbb{S}\cup\mathbbm{j}) \right) \right) \nonumber \\
&=& \min \left(1, \frac{1}{\hat{LC}(\mathbb{S}\cup\mathbbm{\{i,j\}})}\cdot \left( \sum_{k\in S_1} \hat{\alpha}^+_i(\mathbb{S}\cup\mathbbm{\{i,j\}\setminus k}) 
      	+ \hat{\alpha}^+_i(\mathbb{S}\cup\mathbbm{i})+ \hat{\alpha}^+_i(\mathbb{S}\cup\mathbbm{j}) \right)  \right) \nonumber \\
&=&    \hat{\alpha}^+_i(\mathbb{S}\cup\mathbbm{\{i,j\}})   \nonumber
 \end{eqnarray}
 Here the first inequality holds by (\ref{eq:quar1}). The second inequality holds by the induction hypothesis.
 The fourth equality holds since, as we previously argued, $\mathbb{S}\cup\mathbbm{j}$ must be a loyal child.
 It follows that $\hat{\alpha}^+_i(\mathbb{S}\cup\mathbbm{\{i,j\}} ) \le \alpha^+_i(\mathbb{S}\cup\mathbbm{\{i,j\}})$, completing the proof of (\ref{eq:D4}).
\end{proof}

\begin{customlemma}{A.10}\label{lem:quarrel-minus}
Then, for any $S\subseteq [n]$ with $i,j\notin S$, the efficacy scores of player $i$ satisfy:
\begin{align*}
\hat{\alpha}^-_i(\mathbb{S}) &\le \alpha^-_i(\mathbb{S}) \\ 
\hat{\alpha}^-_i(\mathbb{S}\cup\mathbbm{i}) &\le \alpha^-_i(\mathbb{S}\cup\mathbbm{i}) \\
\hat{\alpha}^-_i(\mathbb{S}\cup\mathbbm{j}) &= \alpha^-_i(\mathbb{S}\cup\mathbbm{j}) \\
\hat{\alpha}^-_i(\mathbb{S}\cup\mathbbm{\{i,j\}} ) &= \alpha^-_i(\mathbb{S}\cup\mathbbm{\{i,j\}}) 
\end{align*}
\end{customlemma}
\begin{proof}
Apply a symmetric argument to that in the proof of Lemma~\ref{lem:quarrel-minus}.
\end{proof}

\restatetheorem{thm:quarrel-postulate}
RM satisfies the standard quarrel postulate.
\end{theorem}
\begin{proof}
We have
\begin{equation*}
\hat{RM}_i 
\ =\ \sum_{S\in \hat{\mathcal{D}}} \hat{\alpha}_i(\mathbb{S})  \cdot \hat{\mathbb{P}}(\mathbb{S})
\ =\ \sum_{S\in \mathcal{D}} \hat{\alpha}_i(\mathbb{S})\cdot \mathbb{P}(\mathbb{S})
\ \le\ \sum_{S\in \mathcal{D}} \alpha_i(\mathbb{S})\cdot \mathbb{P}(\mathbb{S})
\ =\  RM_i
\end{equation*}
Here the first inequality holds by Lemma~\ref{lem:quarrel-plus} and Lemma~\ref{lem:quarrel-minus}.
Similarly $RM_i\ge \hat{RM}_i$.
Thus $RM$ satisfies {\sc (quar-1)} and {\sc (quar-2)} for any voting-independent probability distribution.
It follows that it satisfies them for equiprobable divisions, i.e., for $RM'$.
\end{proof}

\end{document}